\documentclass[a4paper,onecolumn,10pt,unpublished]{quantumarticle}
\pdfoutput=1
\usepackage[english]{babel}
\usepackage[T1]{fontenc}
\usepackage{amsmath}
\usepackage{tikz}
\usepackage{lipsum}

\usepackage[numbers,sort&compress]{natbib}

\usepackage{amsthm}
\newtheorem{theorem}{Theorem}
\newtheorem{definition}{Definition}
\newtheorem{lemma}{Lemma}
\newtheorem{corollary}{Corollary}

\usepackage{mathtools}
\usepackage{qcircuit}
\usepackage{cancel}
\usepackage{amssymb}
\usepackage{subcaption}
\usepackage{hyperref}
\allowdisplaybreaks

\newtheorem{assumption}{Assumption}
\DeclareUnicodeCharacter{2009}{\,} % Maps thin space to LaTeX command
\newcommand{\R}{\mathbb{R}}

\newcommand{\Tr}{\mathrm{Tr}}
\renewcommand{\part}[2]{\frac{\partial #1}{\partial #2}}
\newcommand{\Eo}{\mathcal{E}}
\newcommand{\Li}{\mathcal{L}}
\newcommand{\Bc}{\mathcal{B}}
\newcommand{\scriptH}{\mathcal{H}}
\newcommand{\scriptK}{\mathcal{K}}
\newcommand{\bigO}{\mathcal{O}}
\newcommand{\polylog}{\mathrm{polylog}}
\newcommand{\supp}{\mathrm{supp}}
\newcommand{\Prep}{\mathrm{Prep}}
\newcommand{\Sel}{\mathrm{Sel}}

\begin{document}

\title{Parallelizing Dissipative Quantum Algorithms}

\author{Sudatta Hor}
\email{sudatta.hor@yale.edu}
\affiliation{Department of Computer Science, Yale University}
% \orcid{0000-0002-2445-2701}
\author{Nathan Wiebe}
\affiliation{Department of Computer Science, University of Toronto}
\affiliation{Pacific Northwest National Laboratory}
\author{Yongshan Ding}
\affiliation{Department of Computer Science, Yale University}
\maketitle

\begin{abstract}
Inspired by classical Markov Chain Monte Carlo (MCMC) methods, \textit{dissipative quantum algorithms} are an emerging class of methods for preparing quantum states on a quantum computer. They operate by simulating a dissipative process governed by a mathematical operator known as the \textit{Lindbladian}. The performance of these algorithms can be limited by the sequential implementation of the Lindbladian’s constituent \textit{jump operators}, which leads to high circuit depth.
To address this issue, this work analyzes a parallelization scheme to implement jump operators simultaneously. We leverage the structure of geometric locality to show that individual jump operators remain spatially localized. This exponentially reduces the circuit depth per time step, resulting in an overall polynomial improvement. Numerical experiments demonstrate a 53x reduction in query depth for a 1D benchmark system of 100 qubits. The reduced circuit depth makes dissipative algorithms more practical for near-term quantum computing.
\end{abstract}

\section{Introduction}
\label{sec:intro}

% Dissipative algs are emerging
Given the ubiquitous use of Markov Chain Monte Carlo (MCMC) methods in classical algorithms, it is expected that analogous methods for quantum algorithms will have an equally prominent role in quantum computing. Motivated by this, researchers have sought to develop quantum algorithms that take inspiration from classical MCMC techniques \cite{temme_quantum_2011, szegedyQuantumSpeedupMarkov2004, jiangQuantumMetropolisSampling2024}. However, due to fundamental principles of quantum mechanics such as the no-cloning theorem and the uncertainty principle, directly translating classical MCMC methods into a quantum algorithm can be difficult. To circumvent these obstacles, \textit{dissipative quantum algorithms} have emerged as a powerful, general framework for translating classical MCMC methods into their quantum counterparts, showing prominent applications in thermal and ground state preparation \cite{chen_quantum_2023, rall_thermal_2023, ding_single-ancilla_2024, zhan_rapid_2025, chen_efficient_2025, ding_efficient_2025, hahn_efficient_2025, gilyen_quantum_2026, lloyd_quantum_2026, ong_rigorous_2026}. In these algorithms, continuous-time Markov chains are generalized into continuous-time quantum Markov chains governed by the \textit{Lindbladian} formalism, which models an open quantum system interacting with its environment. By simulating the corresponding Lindbladian on a quantum computer, this naturally translates into a quantum algorithm.

However, implementing a dissipative quantum algorithm on a quantum computer may require a large amount of computational resources. Simulating its Lindbladian relies on utilizing intricate block encodings and unitary control circuits that assume fully fault-tolerant quantum computers. For early quantum computers with constrained resources such as limited runtime and qubits, these complications make it difficult to implement such algorithms.

% Single ancilla
As a step toward making dissipative algorithms practical for near-term hardware, \cite{ding_single-ancilla_2024} introduced an approach tailored specifically for ground state preparation. Because targeting the ground state inherently avoids the stringent detailed balance conditions (which balance the heating and cooling transitions required for the Gibbs state to be stationary), the mathematical structure of the Lindbladian is vastly simplified. Additionally, it consists of only a single jump operator, which further reduces its complexity. Together, these simplifications yield an algorithm that can be executed using just one ancillary qubit, which is more suitable for near-term implementation.

% Bulk dissipation and the Depth Problem
While \cite{ding_single-ancilla_2024} presents an implementation that is highly efficient in qubit cost, its time for convergence, i.e. its \textit{mixing time}, can be very large since only a single jump operator is used. To address this, \cite{zhan_rapid_2025} introduced \textit{bulk dissipation}—a technique that applies multiple jump operators across the system. This achieves \textit{rapid mixing}, which means the mixing time grows only logarithmically in the system size. Unfortunately, bulk dissipation increases the circuit depth per simulation step, which introduces a significant computational tradeoff. Using the same simulation method, the circuit depth per step increases linearly with the number of jump operators. In general, this highlights a tradeoff depicted by the following equation.
\begin{center}
    (total circuit depth) $=$ (circuit depth per unit time) $\times$ (mixing time)
\end{center}
Bulk dissipation reduces the mixing time, but the circuit depth per unit time increases which diminishes the reduction in total circuit depth.

% comparison table
\begin{table}
    \centering
    \renewcommand{\arraystretch}{2}
    \begin{tabular}{lccc}
    \hline
    \textbf{Method} & \textbf{Circuit Depth} & \textbf{Mixing Time} & \textbf{Total Depth} \\
    \hline
    Single jump dissipation \cite{ding_single-ancilla_2024} 
        & $\tilde \bigO(n^3 t^2)$
        & $\bigO(n^3 \log n)$
        & $\tilde \bigO(n^9)$ \\
    Sequential bulk dissipation \cite{zhan_rapid_2025}
        & $\tilde \bigO(n^5 t^2)$
        & $\bigO(\log n)$
        & $\tilde{\bigO}(n^5)$ \\

    Parallel bulk dissipation (this work) 
        & $\tilde{\bigO}(n t^2)$
        & $\bigO(\log n)$
        & $\tilde{\bigO}(n)$ \\
    \hline
    \end{tabular}
    \caption{Comparison of query depth for three dissipation methods on a system of $n$ qubits for a fixed precision $\epsilon$. We consider the mixing times for the 1D transverse-field Ising model in the weak interaction regime derived in \cite{zhan_rapid_2025} and use it to compute the total depth required.}
    \label{table:comparison-table}
\end{table}
% \begin{table}
%     \centering
%     \renewcommand{\arraystretch}{1.5}
%     \setlength{\tabcolsep}{20pt}
%     \small
%     \begin{tabular}{lccc}
%         \hline
%         & \textbf{Single jump}
%         & \textbf{Sequential bulk}
%         & \textbf{Parallel bulk} \\
%         \hline
%         Query depth
%         & $\tilde{\bigO}(n^2 t^2)$
%         & $\tilde{\bigO}(n^4 t^2)$
%         & $\tilde{\bigO}(n t^2)$ \\

%         Circuit depth
%         & $\tilde{\bigO}(n^3 t^2)$
%         & $\tilde{\bigO}(n^5 t^2)$
%         & $\tilde{\bigO}(n t^2)$ \\

%         Mixing time
%         & $\bigO(n^3)$
%         & $\bigO(\log n)$
%         & $\bigO(\log n)$ \\

%         Total depth
%         & $\tilde{\bigO}(n^9)$
%         & $\tilde{\bigO}(n^5)$
%         & $\tilde{\bigO}(n)$ \\
%         \hline
%     \end{tabular}
%     \caption{Comparison of Lindbladian simulation methods on a 1D
%     system of $n$ qubits for fixed precision $\epsilon$.
%     Single-jump dissipation is from
%     \cite{ding_single-ancilla_2024}, sequential bulk dissipation
%     is from \cite{zhan_rapid_2025}, and parallel bulk dissipation
%     is this work. The mixing times are based on
%     \cite{zhan_rapid_2025}.}
%     \label{tab:comparison}
% \end{table}

% Contributions
\textbf{Contributions.} In this paper, we parallelize the implementation of jump operators which substantially reduces the total circuit depth. We tailor our method towards application-specific algorithms such as those in \cite{chen_quantum_2023, rall_thermal_2023, ding_single-ancilla_2024, zhan_rapid_2025, chen_efficient_2025, ding_efficient_2025, hahn_efficient_2025, gilyen_quantum_2026, lloyd_quantum_2026, ong_rigorous_2026}. To simulate a quantum system given by a Hamiltonian $H$, we consider a class of Lindbladians of the form
\begin{equation}
    \frac{d\rho}{dt} = \Li [\rho] \coloneq -i [H, \rho] + \sum_a\left( K_a \rho K_a^\dagger - \frac{1}{2} \left\{K_a^\dagger K_a, \rho \right\}\right)
\end{equation}
where the \textit{jump operators} $K_a$ are defined as
\begin{equation}
    K_a \coloneq \int_{-\infty}^\infty f(t) A_a(t) \, dt \, .
\end{equation}
$a$ indexes each system qubit, $A_a$ is a single qubit Pauli operator acting on qubit $a$, and $f(t)$ is a real-valued filter function. We focus on the case where $H$ is \textit{geometrically local}, which is ubiquitous in physics and encompasses a broad class of practically relevant systems \cite{kogut_hamiltonian_1975, haah_quantum_2023}. For these systems, there exists Lieb-Robinson bounds \cite{nachtergaele_locality_2009, poulin_lieb-robinson_2010, hastings_locality_2010, barthel_quasilocality_2012, anthony_chen_speed_2023} which constrain the spread of observables over time, giving rise to a light-cone structure. Using these bounds, we show that each jump operator can be spatially localized to a region whose radius scales logarithmically. As a result, jump operators acting on disjoint 
regions can be executed in parallel, which exponentially reduces the circuit depth per time step.

To simulate the Lindbladian, we utilize a first-order Trotter formula and Stinespring dilation similar to \cite{ding_single-ancilla_2024, hahn_efficient_2025}. While other methods have asymptotically lower circuit depth \cite{childs_efficient_2017, cleve_efficient_2017, li_simulating_2023}, we choose this method because of its simplicity and compatibility towards early fault-tolerant quantum computers with limited resources. Furthermore, our method explicitly constructs the jump operator $K_a$ while other methods only assume access to its block encoding. This gives an end-to-end resource estimate for an application-specific jump operator.

We compare the total circuit depth of three different dissipation methods:
\begin{itemize}
    \item \textit{Single jump dissipation} \cite{ding_single-ancilla_2024}, which uses a single jump operator on one qubit. 
    \item \textit{Sequential bulk dissipation} \cite{zhan_rapid_2025}, which sequentially applies a jump operator on each site without localization.
    \item \textit{Parallel bulk dissipation}, which localizes jump operators on each site to apply them in parallel.
\end{itemize}

In our work, we compare the asymptotic circuit depths for each method. Table~\ref{table:comparison-table} contains an example that shows the circuit depths between each method, accounting for different mixing times. While sequential bulk dissipation has a higher circuit depth compared to the single jump case, the mixing time is exponentially faster which results in a quartic improvement in total circuit depth. Furthermore, our method of localization and parallelization gives an additional quartic improvement in circuit depth.

Gate complexity is studied in our work instead of Hamiltonian simulation time as done in \cite{ding_single-ancilla_2024}. This gives a more descriptive estimate for the implementation cost on early fault-tolerant quantum computers. Crucially, unlike simulation time, gate complexity heavily depends on the system size. For example, to simulate a Hamiltonian of size $n$ for time $t$, under some assumptions\footnote{We assume that the 1-norm of the LCU coefficients of the Hamiltonian is $\bigO(n)$, and that its LCU consists of $\bigO(n)$ unitaries so the corresponding oracles use $\bigO(n)$ gates.} we expect general algorithms \cite{lowOptimalHamiltonianSimulation2017, low_hamiltonian_2019, gilyenQuantumSingularValue2019} to have query complexity $\bigO(nt +\log (1/\epsilon))$ where each query requires $\bigO(n)$ gates, resulting in a total gate complexity of $\bigO(n^2 t + n\log(1/\epsilon))$. Because of this extra dependence on $n$, localization methods have a big impact in reducing gate complexity.

Our approach is numerically validated on the 1D 
transverse-field Ising model (1D TFIM) on 100 qubits. 
Our experiments demonstrate that parallel bulk dissipation gives a 53x reduction in query depth over the sequential baseline. Furthermore, we show that the single jump method requires an infeasibly high query depth due to its large mixing time. This confirms that parallel bulk dissipation provides the best reduction in query depth, making dissipative algorithms more suitable for near-term quantum computers.

% Related work
\textbf{Related work.} Locality in open quantum systems has been studied and utilized in previous work. 
For Lindbladians with local interactions, it has been shown that they exhibit a Lieb-Robinson bound \cite{poulin_lieb-robinson_2010}, which enables local observables to be efficiently simulated \cite{barthel_quasilocality_2012}. Other works utilized locality to show that the system mixes in polynomial or even logarithmic time \cite{rouze_efficient_2025, rouze_optimal_2026, kastoryano_quantum_2016}. Our work does not investigate mixing time, but instead the implementation cost. In particular, \cite{chen_quantum_2023} suggests that local Lindbladians can potentially be parallelized, but does not provide analysis on circuit depth or establish accuracy and mixing time of the parallelized dynamics.

Recent work proposed Lindbladian simulation algorithms with reduced circuit depth \cite{yu_exponentially_2025, mohammadipour_reducing_2025}. While these methods significantly improve circuit depth, they do not consider locality and obtain only the expectations of observables through sampling and postprocessing rather than evolving the full state. In contrast, our work simulates the full state and exploits locality to reduce circuit depth. This approach is motivated by previous Hamiltonian simulation methods that obtain near-optimal gate complexity for local systems \cite{childs_nearly_2019, haah_quantum_2023}.

Furthermore, our work studies application-specific Lindbladians rather than general ones. General purpose Lindbladian simulation algorithms assume oracle access to the Hamiltonian and jump operators, for example through LCU representations \cite{childs_efficient_2017, cleve_efficient_2017, li_simulating_2023}. Instead, our work explicitly constructs jump operators from a given Hamiltonian and analyzes the cost of this construction. The jump operators that we construct follow a specific form that encompasses applications such as thermal and ground state preparation \cite{chen_quantum_2023, rall_thermal_2023, ding_single-ancilla_2024, zhan_rapid_2025, chen_efficient_2025, ding_efficient_2025, hahn_efficient_2025, gilyen_quantum_2026, lloyd_quantum_2026, ong_rigorous_2026}.

Most notably, our localization method is similar to that in \cite{hahn_efficient_2025}. Their work implements Gibbs sampling algorithms by localizing each jump operator in a similar way. However, our work differs in the following.
\begin{itemize}
    \item Instead of focusing on Gibbs sampling, our work generalizes it and encompasses other dissipative algorithms. In particular, we developed a new assumption of the filter function that includes several dissipative algorithms and gives clear estimates on the circuit cost.
    \item Our work focuses on the parallelization capability and circuit-depth scaling for large scaling system sizes. In contrast, their work places more emphasis on near-term implementation, including numerical studies in the presence of hardware noise. Our numerical experiments study systems of up to 100 qubits, compared with the 12-qubit demonstrations done previously.
    \item Their work compiles each local channel using variational circuits. While this compilation method numerically appears to be effective for small demonstrations, there is no theoretical guarantee that this remains accurate for larger systems. Our work instead constructs the jump operator using quadrature methods and decomposes it into Prepare and Select oracles to count the number of gates. This gives a theoretical upper bound on the circuit depth required for implementation.
\end{itemize}

% Outline
\textbf{Outline.} This article is organized as follows. In Section~\ref{sec:dissalgs}, we provide algorithmic background on dissipative quantum algorithms and give examples for preparing Gibbs and ground states. In Section~\ref{sec:mainresult}, we present \textit{parallel bulk dissipation}. We prove that the localized dynamics remain close to the original, resulting in a similar mixing time and stationary state. We analyze the asymptotic circuit depth and compare it between single jump, sequential, and parallel methods. In Section~\ref{sec:numex}, we demonstrate this parallelization strategy numerically on the 1-dimensional transverse-field Ising model.

\section{Dissipative Algorithms}
\label{sec:dissalgs}
Dissipative algorithms consist of constructing a cleverly designed Lindbladian and then simulating its dynamics. In the following, we give the necessary background on the Lindblad dynamics before describing prominent examples of dissipative algorithms that prepare the Gibbs and ground state.

\subsection{Lindblad Dynamics}
\label{subsec:lindblad}
In a closed quantum system that is described by a Hamiltonian $H_{sys}$, recall that the dynamics on a pure quantum state $|\psi_{sys}\rangle$ is given by
\begin{equation}
    \frac{d}{dt}|\psi_{sys}\rangle = -i H_{sys}\, |\psi_{sys}\rangle \,.
    \label{eq:schrodinger}
\end{equation}
If $H$ is time-independent, then the solution is given by
\begin{equation}
    |\psi_{sys}(t)\rangle = e^{-iH_{sys}t} |\psi_{sys}(0)\rangle\, .
    \label{eq:hamiltonian-evolution}
\end{equation}
In an open quantum system, the system may have interaction with the environment, so the Schr\"{o}dinger equation must use a Hamiltonian $H_{env}$ and pure state $|\psi_{env}\rangle$ that describes the entire system and environment. After solving the Schr\"{o}dinger equation for time $t$ to get $|\psi_{env}(t)\rangle$, the environment can be traced out to get a mixed state $\rho_t$ that describes the state of the system agnostic of the environment.

The motivation for the Lindblad master equation is to describe the dynamics of $\rho_t$. Analogous to (Eq.~\ref{eq:schrodinger}) and (\ref{eq:hamiltonian-evolution}), we wish to find a linear operator $\Li$ such that
\begin{equation}
    \frac{d}{dt}\rho = \Li \rho \quad \text{and} \quad \rho_t = e^{\Li t} [\rho_0]\, 
    \label{eq:lindblad-master}
\end{equation}
where $\rho_0$ denotes the initial state at time $t=0$. This is known as the Lindblad master equation, and $\Li$ is known as the Lindbladian with its dynamics described by $e^{\Li t}$. Assuming the dynamics are Markovian and time-homogeneous, and requiring $e^{{\Li}t}$ to be CPTP, the general form of $\Li$ is given by
\begin{equation}
    \Li [\rho] \coloneq \underbrace{-i [H, \rho]}_{\text{coherent}} + \underbrace{\sum_a\left( K_a \rho K_a^\dagger - \frac{1}{2} \left\{K_a^\dagger K_a, \rho \right\}\right)}_{\text{dissipative}}
    \label{eq:lindbladian}
\end{equation}
where $\{K_a\}$ are known as \textit{jump operators}, which are not necessarily Hermitian \cite{breuer_quantum_2007, manzano_short_2020}.

The coherent part is a direct consequence of the Schr\"{o}dinger equation and its dynamics is unitary. The dissipative part describes the interactions with the environment. The $K_a\rho K_a^\dagger$ operation applies a transition to $\rho$, while the $-\frac{1}{2}\{K_a^\dagger K_a, \rho\}$ operation ensures trace preservation in the event that no transition occurs. The dynamics of the dissipative part is not unitary, and requires additional techniques to be simulated on a quantum computer.

% Markov chain interpretation
% ----- Shortened -------
% The dissipative part can also be interpreted as a quantum analogue of classical MCMC, in particular the Metropolis-Hastings algorithm \cite{levin_markov_2017}. This algorithm evolves a discrete-time Markov chain by applying a random jump (transition) $A^a$ with probability $p(a)$ and accepting the new state with probability $\gamma (\omega) \coloneq \min(1, e^{-\beta \omega})$, where $\omega$ denotes the gain in energy. Otherwise, reject the new state and remain in the original. This is described by the stochastic matrix $M$ given by
% \begin{equation}
%     M_{i,j} = \underbrace{\sum_{a} \Bigl( p(a) \gamma(E_i - E_j) A^a_{i,j}\Bigr)}_{\text{accept}} + \underbrace{R_{i,j} \delta_{i,j}}_{\text{reject}}\, 
% \end{equation}
% where $R_{i,j}$ is a diagonal matrix determined by conservation of probability. This can be generalized into a continuous-time Markov chain where the probability distribution after time $t$ is given by $\pi_t = e^{Lt} \pi_0$, where $L$ is a matrix given by
% \begin{equation}
%     L_{i,j} = \sum_{a} p(a) \left( \underbrace{\gamma(E_i - E_j) A^a_{i,j}}_{\text{transition}} - \underbrace{\delta_{i,j} \sum_k \gamma(E_k - E_j) A^a_{k,j}}_{\text{decay}}\right)
% \end{equation}
% where the decay part exists to ensure that $e^{Lt}$ preserves probability. Analogously, the $K_a\rho K_a^\dagger$ term in the Lindbladian acts as the transition which applies the jump $K_a$ while the $\frac{1}{2} \left\{K_a^\dagger K_a, \rho \right\}$ term acts as the decay which preserves probability.
% -------------------------

The dissipative part can be interpreted as a quantum analogue of classical MCMC, in particular the Metropolis–Hastings algorithm \cite{levin_markov_2017}. In the discrete-time setting, one proposes a random transition $A^a$ with probability $p(a)$ and accepts it with probability $\gamma(\omega)=\min(1,e^{-\beta \omega})$, where $\omega$ is the energy change. Otherwise the state is left unchanged. This can be generalized into a continuous-time Markov chain generator $L$ where the probability distribution after time $t$ is given by $\pi_t = e^{Lt} \pi_0$. In this setting, instead of a rejection term, $L$ contains a decay term that represents the rate at which probability flows out of the state. The Lindbladian has the same structure: the first term $K_a \rho K_a^\dagger$ implements the transition (quantum jump), while the anticommutator term $\tfrac{1}{2}\{K_a^\dagger K_a,\rho\}$ provides the corresponding decay.

% Segway into next section
Dissipative algorithms consist of designing a Lindbladian $\Li$ such that simulating the dynamics $e^{\Li t} [\rho]$ on a quantum computer will give us the desired target state. In the following sections, we will give examples of such algorithms.

\subsection{Gibbs State Preparation}
\label{subsec:thermalstateprep}
Given a Hamiltonian $H$, prepare the quantum Gibbs state, i.e. thermal state, which is given by
\begin{equation}
    \rho_\beta \coloneq e^{-\beta H} / Z(\beta)
    \label{eq:gibbs-state}
\end{equation}
where $Z(\beta) = \Tr(e^{-\beta H})$ is known as the partition function.
If we write $H$ in terms of its spectral decomposition $H = \sum_i E_i |\psi_i\rangle \langle \psi_i|$, then Eq.~\ref{eq:gibbs-state} becomes
\begin{equation}
    \rho_\beta = \sum_i \frac{e^{-\beta E_i}}{Z(\beta)} \, |\psi_i\rangle\langle \psi_i| \, .
\end{equation}
In other words, the quantum Gibbs state is simply the mixed state corresponding to the Gibbs distribution over all eigenstates.

Dissipative Gibbs state preparation aims to prepare the quantum Gibbs state on a quantum computer. A canonical example of such an algorithm simulates the \textit{Davies generator}, which describes thermalization in nature where the system is coupled to a heat bath until thermal equilibrium is reached \cite{chen_quantum_2023}. This is described by the following Lindbladian.
\begin{equation}
    \Li_{\text{Davies}} [\rho] = \sum_{\omega\in \Bc(H)} \gamma(\omega) \left(K_\omega \rho K _\omega^\dagger - \frac{1}{2}\{K_\omega^\dagger K_\omega, \rho\}\right)
    \label{eq:davies-gen}
\end{equation}
where the summation is over all \textit{Bohr frequencies} $\omega$ given by
\begin{displaymath}
\Bc (H) \coloneq \{ E_i - E_j \,|\, E_i , E_j \in \mathrm{spec}(H)\}\, .
\end{displaymath}
The transition weight is $\gamma(\omega) \in [0, 1]$ and chosen such that
\begin{displaymath}
    \frac{\gamma(\omega)}{\gamma(-\omega)} = e^{-\beta \omega}\, .
\end{displaymath}
This is done so that quantum detailed balance is satisfied, which is a condition that heating and cooling transitions are balanced so the Gibbs state is \textit{stationary}, i.e. $e^{\Li t} [\rho_\beta] = \rho_\beta$. For example, we can choose $\gamma(\omega) = \min (1, e^{-\beta \omega})$ similar to the Metropolis algorithm, or $\gamma(\omega) = (e^{\beta \omega}+1)^{-1}$ which is synonymous with Glauber dynamics.
The jump operators $\{K_\omega\}_\omega$ are given by
\begin{equation}
    K_\omega \coloneq \sum_{E_i -E_j = \omega} P_{E_i} A P_{E_j}\, .
    \label{eq:davies-jump-projector}
\end{equation}
where $A$ is some matrix that describes transition amplitudes (e.g. Pauli matrices) and $P_{E_i}$ is the projector onto the $E_i$-subspace.
One can think of this operator as a map that sends $|\psi_j\rangle$ to $|\psi_i\rangle$ with some transition amplitude $\langle \psi_i| A |\psi_j\rangle$.
Another way to write this operator is through the Fourier transform of the Heisenberg evolution, given below.
\begin{equation}
    K_\omega \propto \int_{-\infty}^\infty e^{-i\omega t} \underbrace{e^{iHt} A \, e^{-iHt}}_{\coloneq A(t)} \, dt
    \label{eq:davies-jump}
\end{equation}

Modern dissipative Gibbs-state preparation algorithms construct jump operators that use filtered Heisenberg evolutions \cite{ding_efficient_2025, chen_efficient_2025, hahn_efficient_2025, lloyd_quantum_2026, ong_rigorous_2026}. In particular, the jump operator takes the form
\begin{equation}
K = \int_{-\infty}^\infty f(t) A(t) , dt,
\label{eq:davies-jump-weighted}
\end{equation}
where the filter function $f(t)$ is chosen so that the resulting Lindbladian satisfies an appropriate quantum detailed-balance condition.

Dissipative Gibbs state preparation displays two key points. First, Eq.~\ref{eq:davies-jump-projector} shows that the jump operator $K_\omega$ can be thought of as a map from one energy subspace to another with a transition amplitude that satisfies detailed balance. Secondly, the jump operator in Eq.~\ref{eq:davies-jump-weighted} is weighted by a filter function $f(t)$ whose tails decay rapidly. For local Hamiltonians, this results in quasi-locality of the jump operator. This will be described further in Section~\ref{sec:mainresult}.

\subsection{Ground State Preparation}
\label{subsec:groundstateprep}
Inspired by dissipative Gibbs state preparation, \textit{dissipative ground state preparation} was introduced in \cite{ding_single-ancilla_2024} to prepare the ground state of a given Hamiltonian $H$. Instead of modeling after a physical process found in nature, it uses a Monte Carlo approach that maps eigenstates to those with lower energy. Furthermore, its simplicity makes it more suitable for near-term quantum computers.

The Lindbladian $\Li$ in this method is given by
\begin{equation}
    \Li [\rho] = \underbrace{-i[H, \rho]}_{\scriptH [\rho]} + \underbrace{K\rho K^\dagger - \frac{1}{2} \{K^\dagger K, \rho\}}_{\scriptK [\rho]}
    \label{eq:sagsp-lindblad}
\end{equation}
where the jump operator $K$ is defined to be
\begin{equation}
    K \coloneq \int_{-\infty}^\infty f(t) A(t) \, dt
    \label{eq:sagsp-jump}
\end{equation}
where $A$ is a Hermitian operator that represents the coupling between the system and environment. For example, $A$ can be a Pauli operator on one qubit. $f(t)$ is a filter function such that it has exponentially decaying tails and its Fourier transform $\hat f(\omega)$ satisfies
\begin{equation}
    \hat f(\omega) = 0 \quad \text{when }\omega \geq 0\, .
    \label{eq:sagsp-filter-prop}
\end{equation}

If we write the Hamiltonian in its spectral decomposition $H = \sum_i E_i |\psi_i\rangle \langle \psi_i|$, then the jump operator as defined in
Eq.~\ref{eq:sagsp-jump} can be written as
% the Heisenberg evolution can be written as
% $A(t) = e^{iHt}Ae^{-iHt} = \sum_{i,j} e^{i(E_i - E_j)t} |\psi_i\rangle \langle \psi_i| A|\psi_j\rangle \langle \psi_j|$
\begin{equation}
    K = \sum_{i,j} \hat f(\lambda_i - \lambda_j)\, |\psi_i\rangle \langle \psi_i|A |\psi_j\rangle \langle \psi_j| \, .
    \label{eq:sagsp-jump-summation}
\end{equation}
From this, we can see that $K$ maps an eigenstate $|\psi_j\rangle$ to an eigenstate $|\psi_i\rangle$ with amplitude $\hat f(\lambda_i - \lambda_j) \langle \psi_i |A|\psi_j\rangle$. By the property in Eq.~\ref{eq:sagsp-filter-prop}, we see that $|\psi_j\rangle$ maps to $0$ whenever $\lambda_i \geq \lambda_j$. In other words, $K$ allows transitions only to lower energy eigenstates. Thus, if $|\psi_0\rangle $ is the ground state then
\begin{equation}
    K |\psi_0\rangle = \sum_i \hat f(\lambda_i - \lambda_0) |\psi_i\rangle \langle \psi_i| A|\psi_0\rangle\,  = 0.
\end{equation}
As a result, $\Li [|\psi_0\rangle\langle \psi_0|] = 0$ which implies that the ground state is a fixed point of the Lindblad dynamics, i.e.
\begin{equation}
      e^{\Li t} [|\psi_0\rangle \langle \psi_0|] = |\psi_0\rangle \langle \psi_0|\, .
\end{equation}

\begin{figure}[t]
    \centering
    \vspace{1em}
    \includegraphics[width=0.6\linewidth, trim=0 0.2in 0 0]{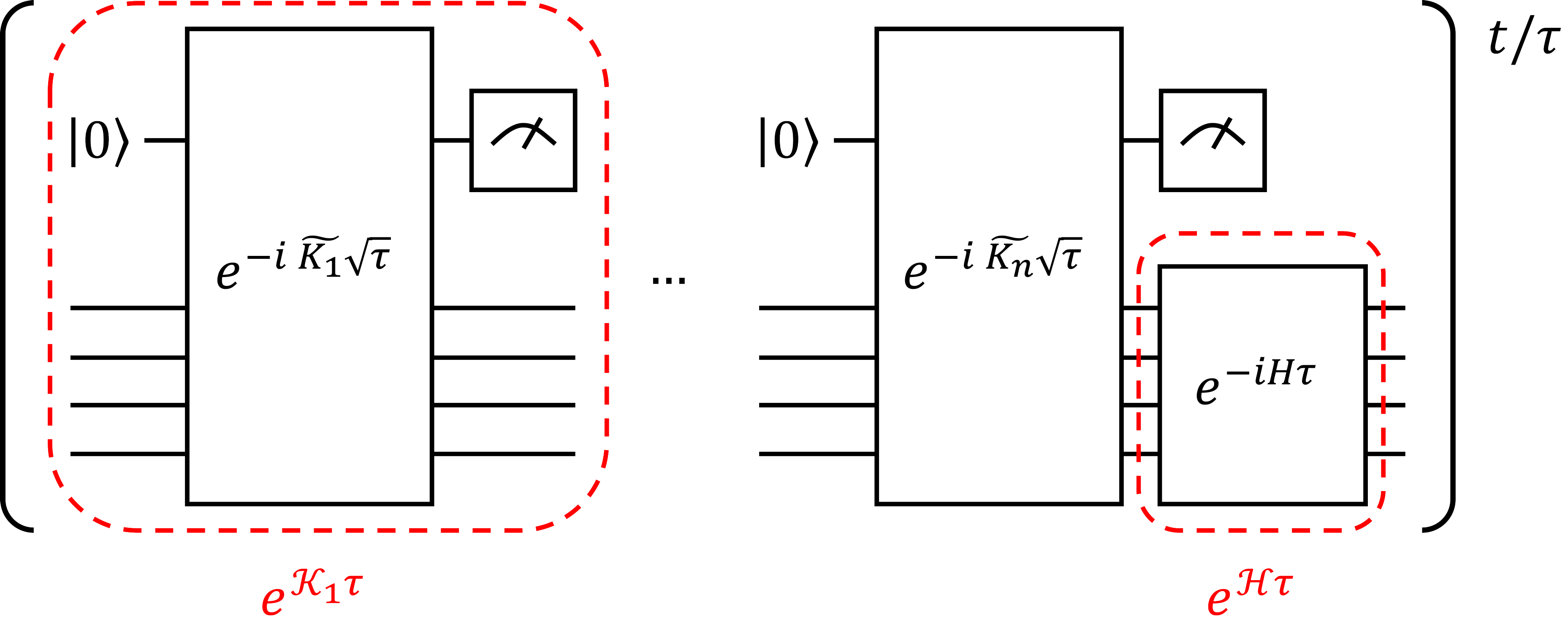}
    \caption{Quantum circuit used to simulate bulk dissipation for preparing the ground state $\rho_0$ of $H$. Given input state $\rho$, the circuit approximates $e^{\Li t}[\rho]$. Each dissipative part $e^{\Li_a\tau}$ is simulated using Stinespring dilation using an ancilla qubit, whose measurement result is discarded. Thus, it is not required to post-select on the measurement outcome.}
    \label{fig:circuit-struct}
\end{figure}

% multiple jump operators
Instead of only using a single jump operator, the Lindbladian can be extended to have multiple, as done in \cite{zhan_rapid_2025}. By doing so, the mixing time dramatically improves. In fact, for certain 1D quantum systems they show that it is possible to have \textit{rapid mixing}, i.e. the mixing time scales logarithmically in the system size. Specifically, the Lindbladian in Eq.~\ref{eq:sagsp-lindblad} is modified so that for each system qubit $j$, the coupling operator $A_a$ is chosen to be a single qubit Pauli operator acting on qubit $j$ and the jump operator $K_a$ is defined as
\begin{equation}
    K_a \coloneq \int_{-\infty}^\infty f(t) A_a(t) \, dt \, .
    \label{eq:rapid-jump}
\end{equation}
The index $a$ and $j$ are separate because there can be multiple $A_a$ for each qubit $j$. Thus there are $\bigO(n)$ many jump operators and the new Lindbladian is given by
\begin{equation}
    \Li [\rho] = -i[H, \rho] + \sum_a K_a\rho K_a^\dagger - \frac{1}{2} \{K_a^\dagger K_a, \rho\}\, .
    \label{eq:rapid-lindblad}
\end{equation}
Choosing jump operators in this way is referred to as \textit{bulk dissipation}.

It is possible to implement bulk dissipation using the single-ancilla method from \cite{ding_single-ancilla_2024}. However, even though the ancilla cost is minimal, the circuit depth can be large.
Since only one ancilla is available, each  jump operator must be applied sequentially to wait for the ancilla to be measured and reset. This linearly increases the circuit depth per time step. Additionally, the Trotter error also increases linearly due to the additional number of terms, which increases the number of time steps required. Altogether, this results in a quadratic increase in circuit depth over the single jump implementation.

\section{Parallel Dissipation}
\label{sec:mainresult}
To reduce circuit depth, we propose utilizing multiple ancilla qubits to apply jump operators in parallel rather than sequentially. This utilizes more qubits, but the circuit depth per time step is reduced exponentially. Furthermore, the total number of qubits remain asymptotically the same. This suggests that the substantial reduction in circuit depth outweighs the increased qubit cost.

\subsection{Setup}
\label{subsec:setup}
% Let $K_a$ be in the form of Eq. \ref{eq:rapid-jump}, and let $\Li$ be defined as in Eq.~\ref{eq:rapid-lindblad}. Our goal is to simulate the dynamics $e^{\Li t}$. We use a first-order Trotter formula with time step $\tau$ given by
% \begin{equation}
%     e^{\Li t} \approx \left(e^{\scriptH \tau}\prod_a e^{\scriptK_a \tau} \right)^{t/\tau}
%     \label{eq:bulk-trotter}
% \end{equation}
% where $\scriptH [\rho] = -i[H, \rho]$ and 
% \begin{equation}
%     \scriptK_a [\rho]= K_a\rho K_a^\dagger - \frac{1}{2} \{K_a^\dagger K_a, \rho\}\, .
% \end{equation}
% The coherent part $e^{\Li_H \tau}$ is unitary and can be implemented using Hamiltonian simulation, but each dissipative term $e^{\Li_a \tau}$ is nonunitary and must be implemented with Stinespring dilation. \rev{add stinespring dilation equation}. The resulting circuit structure is illustrated in Fig.~\ref{fig:circuit-struct}.
% Note that while $\{\scriptK_a\}$ is not necessarily a set of commuting operators, we can apply the product $\prod_a e^{\scriptK_a\tau}$ in any order because the summation $\sum_a \scriptK_a$ can be arranged in any order.

% Setup of geometrically local Hamiltonian
Assume that qubits are laid out in a $D$-dimensional lattice $\Lambda$ of length $L$ in  $\R^D$ and let the Hamiltonian $H$ be geometrically $k$-local. In other words, it can be written as a sum of terms $\sum_X h_X$ such that $h_X$ acts nontrivially only on a region $X$ of radius $k$ and acts as identity elsewhere. If $X$ has a radius larger than $k$ then $h_X = 0$. Furthermore, assume that $\|h_X \| \leq J$ for some constant $J$.

Let each term $h_X$ be given in a linear combination of unitaries (LCU)
\begin{equation}
h_X = \sum_{j=1}^{m_X} \alpha_{X, j} U_{X, j}
\end{equation}
such that the 1-norm of its coefficients and the number of terms $m_X$ is constant. For example, we can write $h_X$ in the Pauli basis. Since $X$ consists of $\bigO(k^D)$ qubits, $h_X$ can be written as the linear combination of $4^{O(k^D
)}$ many Pauli terms. Furthermore, the 1-norm of its coefficients is upper bounded by $2^{\bigO(k^D)} \|h_X\|$, which is constant. 

We assume that we have access to $h_X$ through the oracles $\Prep_{h_X}$ and $\Sel_{h_X}$, defined by
\begin{equation}
    \begin{split}
        \mathrm{Prep}_{h_X}\, |0\rangle  &= \sum_j \sqrt{\frac{|\alpha_{X,j}|}{\alpha_X}}\, |j\rangle\\
        \mathrm{Sel}_{h_X}\, |j\rangle |\psi\rangle &= |j\rangle\, U_{X,j} |\psi\rangle \, .
    \end{split}
\end{equation}
where $\alpha_X \coloneq \sum_{j} |\alpha_{X, j}|$.

Given access to two-qubit gates and single-qubit rotation gates, the $\Prep$ and $\Sel$ oracles can be implemented with a number of gates that scale linearly with the number of terms in its LCU decomposition \cite{low_hamiltonian_2019}.

The $\Prep$ and $\Sel$ oracles for $h_X$ can be used to construct those for larger regions.     Altogether, these oracles are sufficient to construct a circuit that simulates $e^{\Li t}$.

Additionally, we impose the following assumptions on the filter function.

% lemma: local jump operator
\begin{assumption}
    \label{assumption:rapid-decay-filter}
    $f$ decays rapidly. In particular, assume that there exists positive constants $C, \mu$ such that $|f(s)| \leq C e^{-\mu |s|}$ with $\mu > 2c_0$ where $c_0 \coloneq \max_{p \in \Lambda} \sum_{X \ni p} |X|\, \| h_X\| $.
\end{assumption}
\begin{assumption}
    \label{assumption:rapid-decay-filter-fourier}
    $\hat f$ decays rapidly. In particular, assume that there exists some positive constants $\hat C, \hat \mu$ such that $|\hat f(\xi)| \leq \hat C e^{-\hat \mu |\xi|}$.
\end{assumption}
The first assumption is necessary to localize the jump operator so that the Lieb-Robinson bound does not blow up. The second assumption is used in order to apply quadrature of the jump operator with Poisson summation, which is highly effective. If we instead assume that $\hat f$ is compactly supported, we can use the Nyquist-Shannon sampling theorem which is even more effective \cite{ding_single-ancilla_2024}. However, it is not possible to have a compactly supported $\hat f$ while simultaneously satisfying Assumption~\ref{assumption:rapid-decay-filter}. Additionally, it is possible to remove Assumption~\ref{assumption:rapid-decay-filter-fourier} and instead use a quadrature method different than Poisson summation, for example Riemann summation, but this may require larger implementation cost.

These assumptions on $f$ encompass a number of dissipative algorithms such as those in \cite{ding_efficient_2025, chen_efficient_2025, hahn_efficient_2025, lloyd_quantum_2026, ong_rigorous_2026}.

\subsection{Localization}
\label{subsec:localization}
% lemma: Lieb robinson bound
To enable parallelization, we show that each jump operator can be localized and that the resulting dynamics remain close to the original. For geometrically local systems, this can be done by recognizing that the Heisenberg evolution $A_a(t) \coloneq e^{iHt} A_a e^{-iHt}$ in Eq.~\ref{eq:rapid-jump} can be localized using the Lieb-Robinson bound. Using the assumption of geometric locality, this can be done by using the Lieb-Robinson bound, which is described as follows. This is described in the following lemma adapted from \cite{haah_quantum_2023}.
\begin{lemma}[Lieb-Robinson bound]
    Let $A_a$ be an observable on a single qubit $a$, and for some positive integer $r$ let $\Omega$ be the set of qubits within distance $r$ of 
    qubit $a$. Let
    \begin{equation}
        A_a^\Omega (t) \coloneq e^{iH_a^\Omega t} A e^{-iH_a^\Omega t}
    \end{equation}
    where $H_a^\Omega \coloneq \sum_{X\subseteq \Omega} h_X$. Then
    \begin{equation}
        \| A_a(t) - A_a^\Omega (t) \|  \leq  \| A_a\| \, \frac{(2 c_0 |t|)^r}{r!} 
    \end{equation}
    where $c_0 \coloneq \max_{p \in \Lambda} \sum_{X \ni p} |X|\, \| h_X\| $. Here, $\|\cdot \| $ denotes the spectral norm.
\end{lemma}

Fixing the observable $A_a$, this lemma tells us that there exists a constant $v_{LR}$ such that for $r > v_{LR} |t|$, we have 
\begin{equation}
    \|A_a(t) - A_a^\Omega(t)\|   = \bigO (e^{-\mu_{LR} r})
\end{equation}
for some $\mu_{LR} > 0$. The constant $v_{LR}$ is known as the Lieb-Robinson velocity.

% local jump
The Lieb-Robinson bound allows us to localize each jump operator described by the following lemma.

\begin{lemma}
    \label{lemma:appx-jump}
    Define the local operator $K_a^\Omega$ as
    \begin{equation}
    K_a^\Omega \coloneq \int_{-\infty}^{\infty} f(s) A_a^\Omega (s) \, ds \, .
    \end{equation}
    Then for any $\epsilon > 0$, we have $\|K_a - K_a^\Omega \|  < \epsilon$ when
    % \begin{equation}
        $r = \Omega\left(\log \frac{1}{\epsilon} \right)\, .$
    % \end{equation}
\end{lemma}
\begin{proof}
    \begin{align}
        \| K_a - K_a^\Omega \|  &= \left\| \int_{-\infty}^\infty f(s) (A_a(s) - A_a^\Omega(s))\, ds \right\|  \\
        &\leq \int_{-\infty}^\infty |f(s) | \cdot \|A_a(s) - A_a^\Omega (s) \|  \, ds \\
        &\leq 2C \int_0^\infty e^{-\mu s} \cdot \frac{(2c_0 s)^r}{r !} \, ds \\
        % &=2C \frac{(2c_0)^r}{r !} \int_0^\infty e^{-\mu s} s^r \, ds \\
        &= 2C \frac{(2c_0)^r}{r!} \frac{\Gamma (r +1)}{\mu^{r + 1}} 
    \end{align}
    where $\Gamma$ denotes the gamma function. Since $\Gamma(r + 1) = r!$ for positive integers $r$, this bound becomes $\|K_a - K_a^\Omega \|  \leq \frac{2C}{\mu} \left(\frac{2c_0}{\mu}\right)^r$. By assumption, we have $2c_0 / \mu < 1$, so as $r$ increases we have that $\|K_a - K_a^\Omega \|  \to 0$. If $r$ is large enough where $\Omega$ becomes the entire lattice, then we have $K_a = K_a^\Omega$.
    To obtain $\|K_a - K_a^\Omega \|  < \epsilon$, take $\frac{2C}{\mu} \left(\frac{2c_0}{\mu}\right)^r < \epsilon$. Thus, it suffices to take 
    % $r = \Omega \left(\log \frac{1}{\epsilon}\right)$\,.
    \begin{equation}
        r > \frac{\log \frac{1}{\epsilon} + \log \frac{2C}{\mu}}{\log \frac{\mu}{2c_0}} = \Omega \left(\log \frac{1}{\epsilon}\right) \, .
    \end{equation}
    
\end{proof}
Without exponential decay of $f$ according to Assumption \ref{assumption:rapid-decay-filter}, this bound possibly may not hold since the integral bound may not converge. In particular, $\alpha$-Gevrey functions such as those considered in \cite{ding_single-ancilla_2024} are subexponential for $\alpha > 1$ and results in a divergent integral bound when paired with the Lieb-Robinson bound. In this case, tighter bounds than Lieb-Robinson are required for localization.

If we localize every jump operator, the resulting Lindbladian and its dynamics remain close to the original as long as the truncation radius $r$ is sufficiently large. This is shown in the following corollary and theorem.

\begin{corollary}
    \label{cor:appx-lindblad}
    Let $\Li^\Omega$ be the new Lindbladian after replacing each jump operator with its local approximation, i.e.
    \begin{equation}
        \Li^\Omega [\rho] = -i[H, \rho] + \sum_a \bigg(\underbrace{K_a^\Omega \rho (K_a^{\Omega })^{\dagger} - \frac{1}{2} \{ (K_a^\Omega)^\dagger K_a^\Omega, \rho\} }_{\coloneq \scriptK_a^\Omega [\rho]} \bigg)\, .
    \end{equation}
    Then $\left \| \Li -\Li^\Omega \right\|_{1\to1} = \bigO(\epsilon)$ when $r = \Omega\left(\log \frac{n}{\epsilon} \right)$ where $\|\cdot \|_{1\to1}$ denotes the superoperator norm defined as
    \begin{equation}
        \| \Phi\| _{1\to1} = \sup_{X\neq 0} \frac{\|\Phi(X) \|_{1}}{\|X\|_1}\, 
    \end{equation}
    and $\|\cdot \|_1$ denotes the trace norm.   
\end{corollary}
\begin{proof}
   By triangle inequality, we have
    \begin{equation}
        \left\| \Li- \Li^\Omega \right\|_{1\to1} \leq \sum_a \left( \left\| K_a \cdot K_a^\dagger - K_a^\Omega \cdot (K_a^\Omega)^\dagger \right\|_{1\to1} 
        + \frac{1}{2} \left\| \{K_a^\dagger K_a, \cdot \} - \{(K_a^\Omega)^\dagger K_a^\Omega, \cdot \}\right\|_{1\to1} \right)\, .
     \end{equation}
    By Lemma~\ref{lemma:appx-jump} we have $\| K_a - K_a^\Omega\|  < \epsilon / n$. Thus by Lemma~\ref{lemma:jump1} and \ref{lemma:jump2}, we have 
    \begin{equation}
        \left\|\Li -  \Li^\Omega \right\|_{1\to1} \leq \sum_a \bigO\left(\epsilon /n\right)  = \bigO \left(\epsilon\right)\, .
    \end{equation}
\end{proof}

% theorem: overall dynamics is approximated
\begin{theorem}[Localized Lindblad Dynamics]
    \label{thm:approximate-lindblad}
    For any $\epsilon > 0$, we have 
    \begin{equation}
        \left\|e^{\Li t} - e^{\Li^\Omega t} \right\|_{1\to1} \leq \epsilon\qquad  \text{ when }\qquad
        r =  \Omega \left( \log \left(\frac{nt}{\epsilon}\right)\right)\, .
    \end{equation} 
\end{theorem}
\begin{proof}[Proof]
    Lemma~\ref{lemma:duhamel-lindbladian-evol} states that $\left\| e^{\Li t} - e^{\Li^\Omega t} \right\|_{1\to1} \leq \left\|\Li - \Li^\Omega \right\|_{1\to1} t$. To get this to $\bigO(\epsilon)$, it suffices to choose $r = \Omega (\log (nt / \epsilon)) $ by Corollary~\ref{cor:appx-lindblad}.

\end{proof}

Since the localized dynamics are close to the original, we can make the following statement on its mixing time and stationary states.

\begin{theorem}[Stationary States and Mixing Time of Localized Lindbladian]
    \label{thm:approximate-stationary}
    Let $t_{\mathrm{mix}}, t_{\mathrm{mix}}'$ be the mixing times of $\Li, \Li^\Omega$, respectively. Assume that $\Li, \Li^\Omega$ have unique stationary states denoted $\sigma, \sigma'$ respectively. Then for $r = \Omega \left(\log \left(\frac{nt_{\mathrm{mix}} (\epsilon)}{\epsilon}\right)\right)$, we have
    \begin{equation}
        \|\sigma - \sigma' \|_1 = \bigO(\epsilon)
    \end{equation}
    and
    \begin{equation}
        t_{\mathrm{mix}}(4\epsilon) \leq t_{\mathrm{mix}}'
        (\epsilon) \leq  t_{\mathrm{mix}}(\epsilon/4) \, .
    \end{equation}
\end{theorem}

\begin{proof}
    Let $t_{\mathrm{mix}}$ denote $t_{\mathrm{mix}}(\epsilon)$. Then
    \begin{align}
        \left\| \sigma - \sigma' \right\|_1   &= \left\| e^{\Li t_{\mathrm{mix}}} [\sigma] - e^{\Li^\Omega t_{\mathrm{mix}}} [\sigma'] \right\|_1 \\
        &\leq \left\|e^{\Li t_{\mathrm{mix}}} [\sigma] - 
    e^{\Li t_{\mathrm{mix}}}[\sigma'] \right\|_1 + \left\|e^{\Li t_{\mathrm{mix}}}[ \sigma'] - e^{\Li^\Omega t_{\mathrm{mix}}} [ \sigma'] \right\|_1 
    \end{align}
    Each term is bounded by $\epsilon$ by definition of mixing time and Theorem~\ref{thm:approximate-lindblad}. Altogether, it is bounded by $2\epsilon$.

    Let $t = t_{\mathrm{mix}}(\eta)$ for some $\eta > 0$. Then for any density operator $\rho$,
    \begin{align}
        \left \|e^{\Li^\Omega t} [ \rho] - \sigma' \right\|_1 &\leq \left \| e^{\Li^\Omega t}[\rho] - e^{\Li t} [\rho]\right\|_1 + \left\| e^{\Li t}[\rho] - \sigma\right \|_1 + \|\sigma - \sigma'\|_1 \\
        &= \epsilon + \eta + 2\epsilon \\
        &= \eta + 3\epsilon.
    \end{align}
    Choose $\eta = \epsilon$, then by definition of mixing time we have $t_{\mathrm{mix}}' (4 \epsilon) \leq t_{\mathrm{mix}}(\epsilon)$. By symmetry we also have $t_{\mathrm{mix}} (4 \epsilon) \leq t_{\mathrm{mix}}'(\epsilon)$.

\end{proof}

\subsection{Parallelization}
\label{sec:parallelization}
% theorem: overall gate count
Theorem~\ref{thm:approximate-lindblad} and \ref{thm:approximate-stationary} establishes that the radius $r$ need only grow logarithmically in the system size $n$.
In this case, each jump operator acts on a set of $\bigO(r^D)$ local qubits. As shown in Fig.~\ref{fig:parallel-jumps}, jump operators that are $\Omega(r)$ qubits apart can be applied simultaneously.
Thus in one Trotter step, every jump operator can be applied in $\bigO(r^D)$ layers, each containing $\bigO\left( \frac{n}{r^D}\right)$ jump operators. 
The following theorem formalizes this parallelization scheme and quantifies its total circuit depth. We make a distinction between query depth and circuit depth, where the former is the number of layers of queries to $\Prep$ and $\Sel$ and the latter is the number layers of gates.

% % figure: parallel layering
% \begin{figure}
%     \centering
% \textbf{(a)}\\
% \vspace{0.5em}
% \includegraphics[width=0.45\linewidth, trim=0 0.2in 0 0]{fig/sequential-jumps.png}\\
% \vspace{1em}
% \textbf{(b)}\\
% \vspace{0.5em}
% \includegraphics[width=0.35\linewidth, trim=0 0.2in 0 0]{fig/parallel-jumps.png}
%     \caption{Circuit diagrams for $\prod_a e^{\mathcal{L}_a\tau}$, where each block applies the jump operator $K_a$ coupled to qubit $a$. Each block uses an ancilla qubit which is omitted from the figure for clarity. The shading indicates the degree to which each qubit is affected, with darker shading corresponding to a stronger effect and lighter shading to a weaker effect. While illustrated here for a 1D system, the brickwork pattern can be extended to higher dimensions. \textbf{(a)} The sequential implementation, where all blocks can reuse the same ancilla if it is measured and reset after each block. \textbf{(b)} The parallel implementation, applying $k=3$ layers each containing $n/k = 3$ blocks, thus requiring $n/k$ ancilla qubits in total. Although the individual blocks do not commute with one another, their order can be adjusted by reordering the summation $\sum_a \mathcal{L}_a$ prior to Trotterization.}
%     \label{fig:parallel-jumps}
% \end{figure}

\begin{figure}
    \centering
        
    \begin{minipage}{0.6\linewidth}
        \centering
        \textbf{(a)}\\[0.5em]
        \includegraphics[height=3cm]{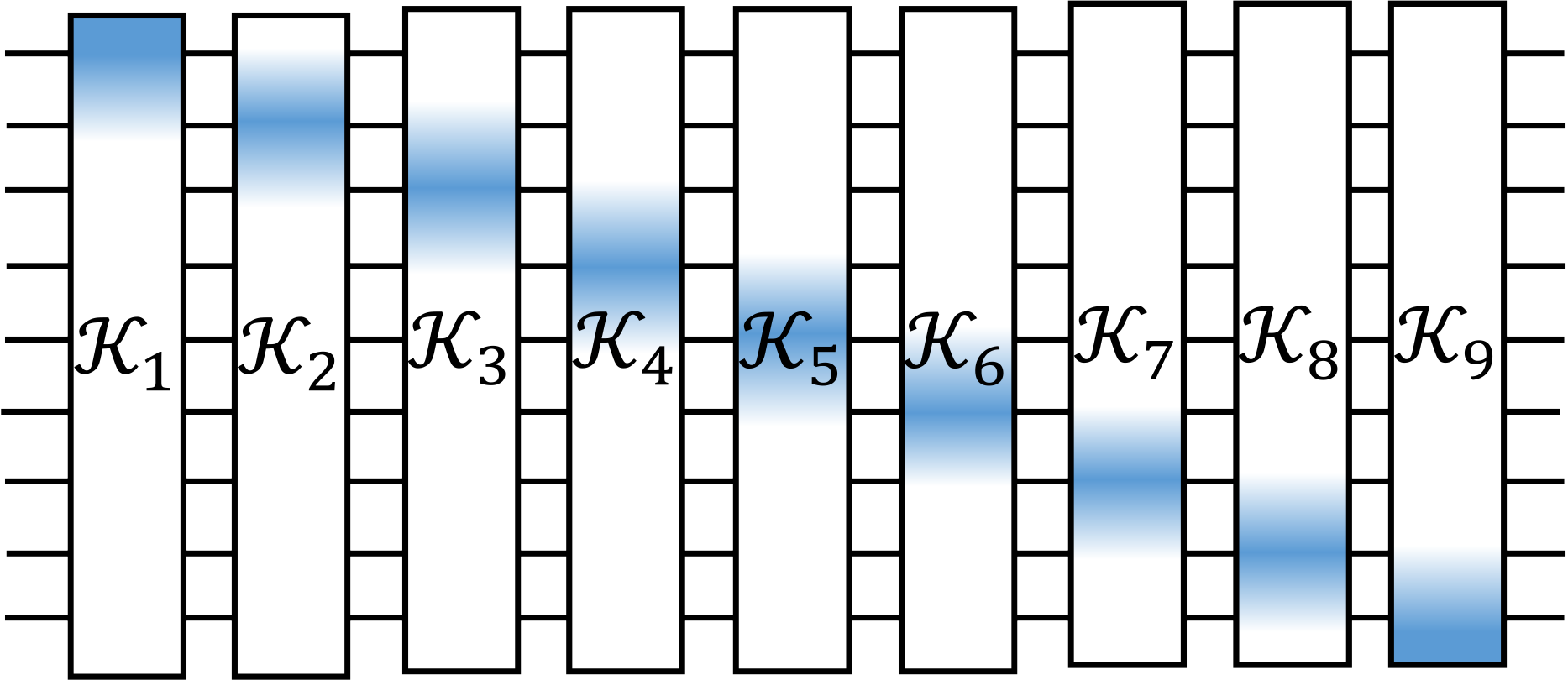}
    \end{minipage}%
    \begin{minipage}{0.4\linewidth}
        \centering
        \textbf{(b)}\\[0.5em]
        \includegraphics[height=3cm]{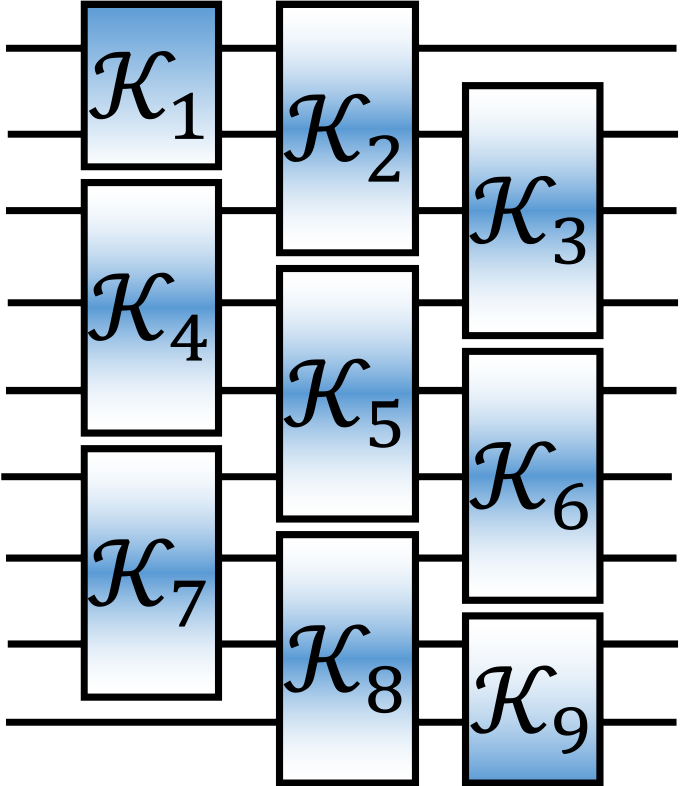}
    \end{minipage}
    
    \caption{Circuit diagrams of $\prod_a e^{\scriptK_a\tau}$ for a 1D system, where each block applies the jump operator $\scriptK_a$ coupled to qubit $a$. Each block uses an ancilla qubit which is omitted for clarity. The shading indicates the degree to which each qubit is affected, with darker shading corresponding to a stronger effect and lighter shading to a weaker effect. \textbf{(a)} The sequential implementation, where all blocks can reuse the same ancilla if it is measured and reset after each block. \textbf{(b)} The parallel implementation. Each layer consists of 3 blocks, thus 3 ancilla qubits are required to apply them simultaneously. Although the individual blocks do not commute with one another, their order can be adjusted by reordering the summation $\sum_a \scriptK_a$ prior to Trotterization.}
    \label{fig:parallel-jumps}
\end{figure}
 
\begin{theorem}[Parallel Bulk Dissipation]
% \textcolor{blue}{NW: The term simulation is vague here, I think you mean a more specific statement about the 1-1 norm.}
    Given $\bigO\left(\frac{n}{\polylog (nt / \epsilon)}\right)$ ancilla qubits that can be measured and discarded, $e^{\Li t}$ can be approximated to precision $\epsilon$ in superoperator 1-1 norm with a query and circuit depth of
    \begin{equation}
    \bigO \left( \|f\|_{L^1}^4 \left( \frac{J}{\mu} + \frac{1}{\mu \hat \mu}\right) \frac{nt^2}{\epsilon} \polylog \left( \frac{nt}{\epsilon}\right)\right)\, 
    \end{equation}
    where we hide polylogarithmic factors in $C, \mu, \hat C, \hat \mu$.
\end{theorem}

\begin{proof}
    Fix $\epsilon > 0$ and approximate $\Li$ with the localized $\Li^\Omega$, then write
    \begin{equation}
        \Li^\Omega  = \sum_{X\subseteq \Lambda} \scriptH_X + \sum_{a} \scriptK_a^\Omega
    \end{equation}
    where $\scriptH_X [\rho] \coloneq -i[h_X, \rho]$. Choose the localization radius as $r = \Theta(\log (nt/\epsilon))$ so that the localization error is $\bigO(\epsilon)$ by Theorem~\ref{thm:approximate-lindblad}.
    
    Consider a small time step $\tau$, then use the first-order Trotter formula which is given by
    \begin{equation}
        \bigg\| e^{\Li^\Omega \tau} - \underbrace{\prod_{X \subseteq \Lambda} e^{\scriptH_X \tau}\prod_{a} e^{\scriptK_a^\Omega \tau}}_{\coloneq S(\tau)} \bigg\|_{1\to1} = \bigO \left(  \left\| \Li^\Omega_{\mathrm{comm}}\right\|\tau^2 \right)
        \label{eq:first-order-trotter-error}
    \end{equation}
    where
    \begin{equation}
    \Li^\Omega_{\mathrm{comm}} \coloneq \sum_{\{X, Y\}} [\scriptH_X, \scriptH_Y] + \sum_{X, a} [\scriptH_X, \scriptK_a^\Omega] + \sum_{a<b} [\scriptK_a^\Omega, \scriptK_b^\Omega]\, .
    \label{eq:commutator-error}
    \end{equation}

    The dissipative terms $e^{\scriptK^\Omega_a \tau}$ are implemented by Stinespring dilation with an ancilla qubit set to $|0\rangle$, given by
    \begin{equation}
        \bigg\|e^{\scriptK_a^\Omega \tau } [\rho] - \underbrace{\Tr_a \left(e^{-i\tilde K_a^\Omega \sqrt{\tau}}\, [|0\rangle \langle 0| \otimes \rho]\, e^{i \tilde K_a^\Omega \sqrt{\tau}}\right)}_{\coloneq \mathcal{E}_a[\rho]}\bigg\|_{1\to 1} = \bigO \left(\|K_a^\Omega \|^4\,  \tau^2\right)
    \end{equation}
    where $\tilde K_a^\Omega$ is the dilated unitary of $K_a^\Omega$ given by
    \begin{equation}
        \tilde K_a^\Omega\coloneq \begin{pmatrix}
            0 & (K_a ^\Omega)^\dagger \\
            K_a^\Omega & 0
        \end{pmatrix}
    \end{equation}

    Let $\|\cdot\|_1$ be shorthand for the superoperator 1-1 norm $\| \cdot \|_{1\to1}$. Then the implementation for one Trotter step is given by
    \begin{equation}
        \bigg\| e^{\Li^\Omega \tau} - \tilde S(\tau) \bigg\|_1  \leq \left\| e^{\Li^\Omega \tau} - S(\tau)\right\|_1 + \left\| S(\tau) - \tilde S(\tau) \right\|_1\, .
    \end{equation}
    where $\tilde S(\tau)$ is the implementation of $\prod_{X \subseteq \Lambda} e^{\scriptH_X \tau}\prod_{a} \Eo_a$ into Prep and Sel oracles as described in Appendix~\ref{sec:gate-implementation}.

    The first term is the error given in Eq.~\ref{eq:first-order-trotter-error}. It can be bounded by computing the commutator error $\|\Li_{\mathrm{comm}}^\Omega\|$. Because each $\scriptH_X$ and $\scriptK_a^\Omega$ are localized, only the pairwise intersecting terms in Eq.~\ref{eq:commutator-error} are nonzero. Together with the triangle inequality, this gives us
    \begin{equation}
        \| \Li_{\mathrm{comm}}^\Omega\| = \bigO\left( nk^D \|h_X\|^2 + n\left(k+r\right)^D \|h_X\| \left\|K_a^\Omega\right\|^2 + nr^D \left\|K_a^\Omega \right\|^4\right)
    \end{equation}
    where we utilized Lemma~\ref{lemma:commutator-bound}. To simplify notation we assumed that each term is uniform that $\|h_X\|$ is the same for all $X$ and $\|K_a^\Omega\|$ is the same for all $a$, but this assumption can be removed by bounding by the supremum. By Lemma~\ref{lemma:jump-operator-norm-bound}, we have $\|\Li_{\mathrm{comm}}^\Omega\| = \bigO(nr^D \|f\|_{L^1}^4)$.
    
    The second term is the error from Stinespring dilation and query implementation of the jump operators. By Lemma~3 in \cite{ding_single-ancilla_2024}, we have that
    \begin{equation}
        \left\| \mathcal{E}_a[\rho] - e^{\scriptK_a^\Omega\tau}[\rho]  \right \|_1 = \bigO\left(\left\|K_a^\Omega \right\|^4\tau^2\right)  = \bigO \left( \|f \|_{L^1}^4 \tau^2\right)\, .
    \end{equation}
    Using Lemma~\ref{lemma:gates-h} and Lemma~\ref{lemma:stinespring-dilation-gates} to approximate $e^{\scriptH_X \tau}$ and $\Eo_a$ with gates, we have that $S(\tau)$ is approximated to accuracy
    \begin{equation}
        \|S(\tau) - \tilde S(\tau)\|_1 = \bigO \left(\|f\|_{L^1} \epsilon' n\tau + \|f\|_{L^1}^4 n\tau^2\right) \, 
    \end{equation}
    where $\epsilon'$ is the quadrature error in Lemma~\ref{lemma:quadrature-error}.
    
    Altogether, the error for one simulation step is
    \begin{equation}
         \bigO\left( \|f\|_{L^1}^4 n r^D \tau^2 + \|f\|_{L^1} \epsilon' n\tau\right) \, .
    \end{equation}
    Repeating this for $t/ \tau$ steps, we have that $e^{\Li^\Omega t}$ is approximated to accuracy
    \begin{equation}
        \frac{t}{\tau} \cdot \bigO\left( \|f\|_{L^1}^4 n r^D \tau^2 + \|f\|_{L^1} \epsilon' n\tau\right) = \bigO(\|f\|_{L^1}^4 n r^D t \tau + \|f\|_{L^1} nt \epsilon ') \, .
    \end{equation}
    % \begin{equation}
    % \begin{split}
    %     \left\| e^{\Li^\Omega t} - \tilde S(\tau)^{t/ \tau} \right \|_1 &= \frac{t}{\tau} \cdot \bigO\left( \|f\|_{L^1}^4 n r^D \tau^2\right) \\
    %     &= \bigO\left( \|f\|_{L^1}^4 n r^D t \tau\right)\, .
    % \end{split}
    % \end{equation}
    For this error to be $\bigO(\epsilon)$, it suffices to choose
    \begin{equation}
        \tau = \bigO \left( \frac{\epsilon}{nr^D t} \frac{1}{\|f\|_{L^1}^4}\right) \,, \quad \epsilon' = \bigO\left( \frac{\epsilon}{nt} \frac{1}{\|f\|_{L^1}}\right)\, .
        \label{eq:tau-formula}
    \end{equation}

    By Lemma~\ref{lemma:stinespring-dilation-gates}, to approximate each channel $\Eo_a[\rho]$ we use
    \begin{equation}
        N_{\mathrm{queries}} = \bigO\left(\left( \frac{J}{\mu} + \frac{1}{\mu \hat \mu}  \right) r^{D} \polylog \left(\frac{nt}{\epsilon }\right) \right)\, 
    \end{equation}
    many queries to $\Prep_{H_a^\Omega}, \Sel_{H_a^\Omega}$, where we hide logarithmic factors in $C, \hat C, \mu, \hat \mu$.

    In each Trotter step, each $\Eo_a[\rho]$ can be implemented simultaneously if they are at least a distance $\Omega(r)$ apart. Along one dimension, $\bigO(L/r)$ of them can be implemented simultaneously and for $D$ dimensions this becomes $\bigO(L^D / r^D)$. Since there are $\bigO(n)$ many terms, $\bigO(r^D)$ layers are required. Similarly each $e^{\scriptH_X\tau}$ can be implemented in $\bigO(k) = \bigO(1)$ layers. Overall, each Trotter step has query depth 
    \begin{equation}
    \bigO(r^D) \times N_{\mathrm{queries}} = \bigO\left(\left( \frac{J}{\mu} + \frac{1}{\mu \hat \mu}\right) r^{2D}\polylog \left(\frac{nt}{\epsilon}\right)\right) \, .
    \label{eq:circuit-depth-per-step}
    \end{equation}
    
    Thus, the total query depth is 
    \begin{equation}
        \begin{split}
        (\text{\# steps}) \times (\text{query depth per step}) = \frac{t}{\tau} \times \bigO\left(\left( \frac{J}{\mu} + \frac{1}{\mu \hat \mu}\right) r^{2D}\polylog \left(\frac{nt}{\epsilon}\right)\right)\\
        = \bigO\left(\frac{n r^D t^2}{\epsilon} \|f\|_{L^1}^4 \right) \times \bigO\left(\left( \frac{J}{\mu} + \frac{1}{\mu \hat \mu}\right)r^{2D}\polylog \left(\frac{nt}{\epsilon}\right)\right)\, \\
        = \bigO\left( \|f\|_{L^1}^4 \left( \frac{J}{\mu} + \frac{1}{\mu \hat \mu}\right)\frac{n t^2}{\epsilon}r^{3D}\polylog \left(\frac{nt}{\epsilon}\right)\right)\,
        \end{split}
        \label{eq:circuit-depth-count-parallel}
    \end{equation}
    Since each oracle consists of $\bigO(r^D)$ gates, we have the circuit depth is
    \begin{equation}
        \bigO\left( \|f\|_{L^1}^4 \left( \frac{J}{\mu} + \frac{1}{\mu \hat \mu}\right)\frac{n t^2}{\epsilon}r^{4D}\polylog \left(\frac{nt}{\epsilon}\right)\right)\, .
    \end{equation}
    By our choice of $r$, we obtain the desired result.

    For the number of qubits, we require one ancilla qubit for each jump operator that is applied simultaneously. After each layer, the ancilla qubits can be reset to be used for the next. In one layer, $\bigO(n /r^D)$ jump operators are applied simultaneously so by our choice of $r$, we obtain the desired result.

\end{proof}

% note on mixing time improvement
Without localization and parallelization, the sequential method requires significantly more time steps and larger circuit depth for each step, the former of which is due to a larger commutator error compared in Eq.~\ref{eq:commutator-error} since there are no longer pairwise terms that are zero. Taking the radius $r$ to cover the entire lattice ($r^D = \bigO(n)$), this altogether results in a circuit depth of
\begin{equation}
\bigO\left( \|f\|_{L^1}^4 \left( \frac{J}{\mu} + \frac{1}{\mu \hat \mu}\right)\frac{n^5 t^2}{\epsilon} \polylog \left(\frac{nt}{\epsilon}\right)\right)\, .
\end{equation}
If we hide polylogarithmic factors and hold factors other than $n,t$ to be constant, this circuit depth scales as $\tilde \bigO(n^5t^2)$ whereas that of the parallel method is $\tilde \bigO (n t^2)$.

The single jump implementation from \cite{ding_single-ancilla_2024} requires less time steps and less circuit depth per step since the commutator error no longer has an additional $\bigO(n)$ dependence and the number of layers per step is only $\bigO(1)$. The circuit depth in this method turns out to be
\begin{equation}
    \begin{split}
    (\text{\# steps}) \times (\text{circuit depth per step}) = \frac{t}{\tau} \times \bigO\left(\left( \frac{J}{\mu} + \frac{1}{\mu \hat \mu}\right)n^2 \polylog \left(\frac{1}{\tau^2}\right)\right)\\
    = \bigO\left(\frac{t^2}{\epsilon} \|f\|_{L^1}^2 \left(nJ + \|f\|_{L^1}^2\right) \right) \times \bigO\left(\left( \frac{J}{\mu} + \frac{1}{\mu \hat \mu}\right) n^2 \polylog \left(\frac{nt}{\epsilon}\right)\right)\\
    = \bigO\left( \|f\|_{L^1}^2 \left(nJ + \|f\|_{L^1}^2\right) \left( \frac{J}{\mu} + \frac{1}{\mu \hat \mu}\right)\frac{n^2 t^2}{\epsilon} \polylog \left(\frac{nt}{\epsilon}\right)\right)\, .
    \end{split}
\end{equation}
Hiding polylogarithmic factors and holding factors other than $n,t$ to be constant, this circuit depth scales as $\tilde \bigO(n^3t^2)$.

\cite{zhan_rapid_2025} shows that for some 1D systems, bulk dissipation exhibits a mixing time of $\bigO(\log n)$ while boundary dissipation (one jump operator on each end) exhibits a mixing time of $\bigO(n^3)$. Using this as a heuristic, the single jump dissipation has depth $\tilde \bigO (n^9)$, sequential bulk dissipation scales as $\tilde \bigO(n^5)$ and its parallelized version scales as $\tilde \bigO(n)$. This example suggests that parallel bulk dissipation has the lowest depth to prepare the ground state.

% set up next section
In the next section, we show a concrete numerical example where the parallelization method displays a substantial reduction in depth when preparing the ground state of a $n=100$ site 1D transverse-field Ising model.

\section{Numerical Example}
\label{sec:numex}
In our numerical experiments, we demonstrate parallel bulk dissipation on the 1D transverse-field Ising model (1D TFIM) on $n=100$ spin sites. The Hamiltonian for this system is given by
\begin{equation}
    H = -g \sum_{i=1}^n Z_i - J \sum^{n-1}_{i=1}X_i X_{i+1}\, 
\end{equation}
where $g$ and $J$ are some constants and $Z_i, X_i$ denote the Pauli operators acting on the $i^\text{th}$ site. We choose $g = 1.5$ and $J = 1.0$.

% fig: filter function
\begin{figure}
  \centering
  \includegraphics[width=\textwidth]{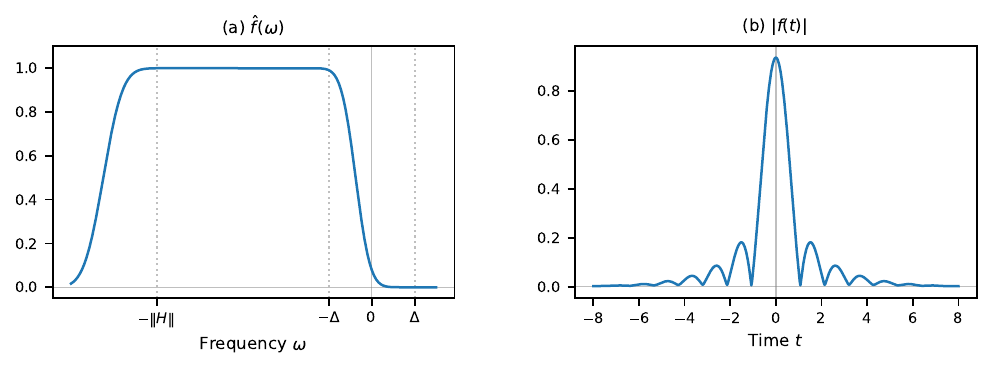}
  \caption{Filter function used in our numerical experiments. Left: The Fourier transform of the filter function $\hat f(\omega)$ in frequency space. It goes to $0$ for positive $\omega$ to ensure that the jump operator maps to lower energy eigenstates. $\hat f(\omega)$ starts to decay at $-\Delta$ so that any energy difference has a strong transition amplitude. Right: The magnitude of the filter function $|f(t)|$ in time space. Its tails decay rapidly as $|t| \to \infty$.}
  \label{fig:filter-functions}
\end{figure}

We choose a filter function $\hat f(\omega)$ defined as
\begin{equation}
    \hat f(\omega) \coloneq \frac{1}{2}\left(  \mathrm{erf}\left(\frac{\omega + a}{\delta_a} \right) - \mathrm{erf}\left( \frac{\omega + b}{\delta_b}\right)\right) \,
    \label{eq:numerical-filter}
\end{equation}
where $\mathrm{erf}(\omega) = \frac{2}{\sqrt{\pi}}\int_0^\omega e^{-x^2}\, dx$ is the error function. The four parameters are chosen to be $a = 1.25 \|H \| , \, \delta_a = 0. 1 \|H \| ,\,  b = \delta_b = 0.375 \Delta$ where $\Delta$ is the spectral gap of the $H$. For our system, $\Delta = 1.0$. Its inverse Fourier transform $f(t)$ can be written as
\begin{equation}
    f(t) = \frac{\exp\left(-\frac{\delta_a^2 t^2}{4} + i at\right) - \exp \left( \frac{-\delta_b^2 t^2}{4} + ibt\right) }{2\pi i t}\, .
\end{equation}
$\hat f(\omega)$ and $|f(t)|$ are plotted in Fig.~\ref{fig:filter-functions}. Note that while Eq.~\ref{eq:sagsp-filter-prop} requires $\hat f(\omega) = 0$ for $\omega > 0$, this function does not satisfy that exactly but instead $\hat f(\omega)$ rapidly approaches $0$. This still yields good mixing time numerically.

According to Lemma~\ref{lemma:appx-jump} the locality of each jump operator depends on the decay rate of $|f(t)|$, which depends on the parameter $\delta_b$. For our choice of parameters, $|f(t)|$ does not decay fast enough for the Lieb-Robinson bound to be feasible. However, this bound is very loose and the localization can be made smaller
% , see Fig.~\ref{fig:local-dynamics-comparison}a
. In practice, a support of radius $r = 16$ closely approximates the jump operator, while the Lemma~\ref{lemma:appx-jump} requires $r > 100$. Although the bound in Lemma~\ref{lemma:appx-jump} can be improved by choosing different parameters such that $|f(t)|$ decays faster, this results in a much slower mixing time, see Appendix~\ref{appendix:filter}. Furthermore, our parameters allow us to choose the radius as low as $r = 4$ and still obtain a close approximation to the original dynamics
 % as shown in Fig.~\ref{fig:local-dynamics-comparison}b
.

Using the computational methods for quasi-free systems in \cite{zhan_rapid_2025}, we simulate the Lindblad dynamics of $e^{\Li t} [\rho]$ where $\rho = \mathbb{I} / 2^n$ is the maximally mixed state. We run the simulation for bulk dissipation where the set of coupling operators $\{A_a\}$ are chosen to be $X_i, Y_i$ for $i=1, \dots, n-1$, for boundary dissipation where the set of coupling operators is $\{X_n, Y_1\}$, and for dissipation with a single coupling operator $X_{50}$ in the middle of the chain. We compute the energy $E_t$ of the system after time $t$ and plot it against the theoretical upper bound on query depth given by Eq.~\ref{eq:circuit-depth-count-parallel} (including constant factors) for fixed $\epsilon = 10^{-3}$. For bulk dissipation, the simulation is repeated for varying levels of parallelization for $r = 4, 10, 16$. We observe that lower values of $r$ require less query depth but limits the accuracy of $E_t$. For single and boundary dissipation, while the query depth per time step is much lower, the substantially large mixing time makes it infeasible.

% % fig: local dynamics comparison
% \begin{figure}
%   \centering
%   \begin{subfigure}{0.4\textwidth}
%     \centering
%     \includegraphics[width=\linewidth, trim=0 0.2in 0 0]{fig/plot_lieb_robinson.pdf}
%   \end{subfigure}
%   \begin{subfigure}{0.4\textwidth}
%     \centering
%     \includegraphics[width=\linewidth, trim=0 0.2in 0 0]{fig/local_lindblad_comparison.pdf}
%   \end{subfigure}
%   \caption{
%     Left: The jump operator $K_a$ is approximated by the localized jump operator $K_a^\Omega$. The spectral norm of their difference is $\|K_a - K_a^\Omega\| $ and is plotted against the radius $r$ of the localized support $\Omega$. The dotted line shows the bound given by Lieb-Robinson, which appears to grow up to $r = 50$. For larger $r$, this bound will start to decrease rapidly, but will be impractical at that point. Fortunately, the exact difference is much lower, reaching an error of $\epsilon = 10^{-4}$ when $r \approx 16$. 
%     Right: For $r = 4, 10, 16$, we simulate the Lindblad dynamics $e^{\Li^\Omega t}$ up to time $t = 10$. We can see that the localized dynamics closely matches the exact dynamics. \rev{For smaller values of $r$, it closely matches up to some time at which the localization error becomes large.}
%   }
%   \label{fig:local-dynamics-comparison}
% \end{figure}

\begin{figure}
  \centering
  \includegraphics[width=0.4\textwidth, trim=0 0.2in 0 0]{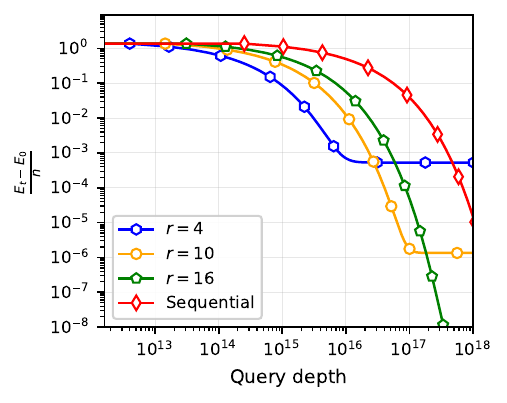}
  \hspace{2em}
  \includegraphics[width=0.4\textwidth, trim=0 0.2in 0 0]{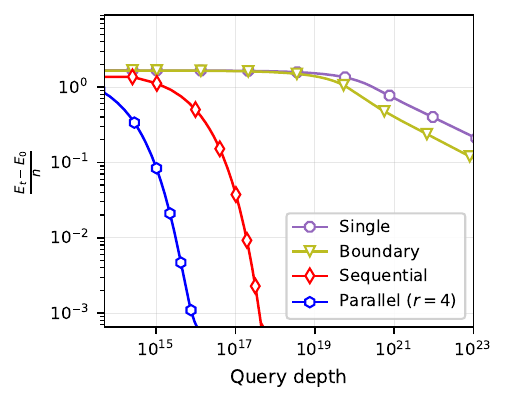}
  \caption{Theoretical upper bounds on query depth. Left: Comparisons on different localization radius $r = 4, 10, 16$. For relative energy error $10^{-3}$, the $r = 4$ implementation has query depth that is 53x lower than that in the sequential implementation. Notice that the stationary state deviates for lower values of $r$.
  Right: Query depth comparisons for single, boundary, and bulk dissipation. While the single and boundary dissipation methods have a lower circuit cost per simulation step, its mixing time is so large that the total cost is infeasible. Parallel bulk dissipation has the lowest query depth out of all of the implementations.}
  \label{fig:circuit-depth-comp}
\end{figure}

\section{Discussion}
\label{sec:disc}
This work demonstrates that parallelization is a powerful tool for running dissipative quantum algorithms on quantum computers. Previously, not much work has utilized parallelization for their implementation. Existing schemes have obtained minimal ancilla cost \cite{ding_single-ancilla_2024}, but they do not consider the increased circuit cost when utilizing bulk dissipation to reduce mixing time. The method in our paper showed that both circuit cost and mixing time can be reduced when using bulk dissipation.

Our method encompasses a broad class of dissipative algorithms whose jump operator has the form of a filtered Heisenberg evolution. We presented a few prominent examples of such instances and imposed conditions on the filter function that allow for parallelization. In particular, we require the rapid decay of the filter function. While this property is commonly used to prove mixing time results, it is also essential for localization when parallelizing the circuit implementation. Our assumption on the filter function is broad and assumes a bound according to $|f(t)| \leq C e^{-\mu |t|}$ for some $C, \mu$. It is important to note that these parameters hide some complexity of the problem. For example, in some instances they may depend on the spectral gap. For provably hard problems with a small spectral gap, the hardness will lie in the filter function which may a slow decay rate.

For sufficiently fast decaying filter function, our theoretical results demonstrated that each jump operator can be localized to arbitrary precision using the Lieb-Robinson bound. We showed that this results in a Lindbladian whose dynamics and stationary state are close to its original. This allows us to use parallelization to simulate the dynamics, which results in an asymptotically lower circuit depth. In particular, we obtain a quartic improvement in circuit depth over the sequential case. This improvement lies in three factors. First, the commutator error of the first-order Trotter formula is reduced since many terms become zero due to localization. Second, simulating the evolution of localized Hamiltonians is much cheaper. Finally, parallelization allows us to apply $\bigO(r^D)$ jump operators simultaneously, reducing the number of sequential operations.

The numerical results demonstrate a 53x reduction in query depth compared to the sequential implementation for ground state energy estimation of the 1D TFIM on $n=100$ sites. This reduction is expected to improve as the system size grows. While our numerical results focus on 1D systems, the parallelization can extend to higher dimensions. Numerical experiments on other systems for larger dimensions and other dissipative algorithms (such as Gibbs state preparation) would be enlightening, but is difficult for large system sizes $n\sim100$ where parallelization is most advantageous.

Although the circuit costs required by our numerical example remain beyond the reach of current quantum hardware, this work suggests that parallelization can substantially reduce these costs, bringing practical dissipative quantum algorithms closer to near-term implementation.

\section*{Author Contributions}
SH wrote the article and developed the results, while NW and YD gave guidance, contributed to the ideas and reviewed the work. AI tools were used for grammar checking, assistance in phrasing for clarity, and to generate code for plotting data.

\begin{acknowledgments}
This material is based upon work supported by U.S. Department of Energy, Office of Science, National Quantum Information Science Research Center, Co-design Center for Quantum Advantage (C2QA) under Contract No. DE-SC0012704. NW acknowledges support from Pacific Northwest National Laboratory's Quantum Algorithms and Architecture for Domain Science (QuAADS) Laboratory Directed Research and Development (LDRD) Initiative as well as   NSERC support throught the Quantum Software Consortium. YD acknowledges partial support by the National Science Foundation (under awards CCF-2312754 and CCF-2338063), by Quantum CT (under NSF Engines award ITE-2302908), by Air Force Office of Scientific Research MURI (FA9550-26-1-B036), by Boehringer Ingelheim, by NSF NQVL-ERASE (under award OSI-2435244), and by NSF QLCI Practiqal (under award OSI-2552449). External interest disclosure: YD is a consultant and equity holder of Toric Labs and D-Wave Quantum.
\end{acknowledgments}

% \balance
% \nocite{*}
\bibliographystyle{quantum}
\bibliography{references-zotero,adref}
\pagebreak

\appendix

\section{Filter function}
\label{appendix:filter}
\begin{figure}
    \centering
    \includegraphics[width=0.9\textwidth]{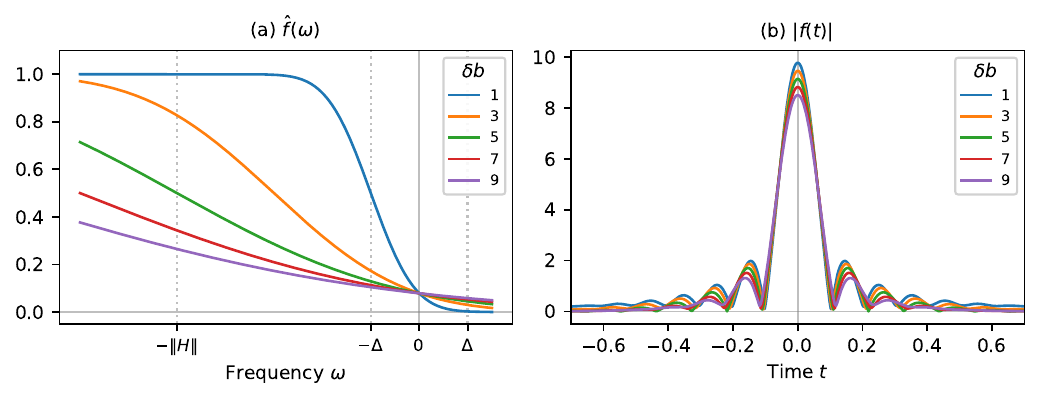}
    
    \caption{Plot of different filter functions with varying $\delta_b$ parameter. Notice that larger values of $\delta b$ give faster decay of $|f(t)|$ but }
    \label{fig:various-filter-plots}
\end{figure}

\begin{figure}
    \centering
        
    \begin{minipage}{0.5\linewidth}
        \centering
        \textbf{(a)}\\[0.5em]
        \includegraphics[height=5cm]{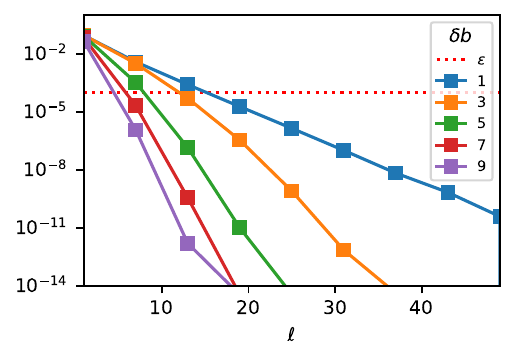}
    \end{minipage}%
    \begin{minipage}{0.5\linewidth}
        \centering
        \textbf{(b)}\\[0.5em]
        \includegraphics[height=5cm]{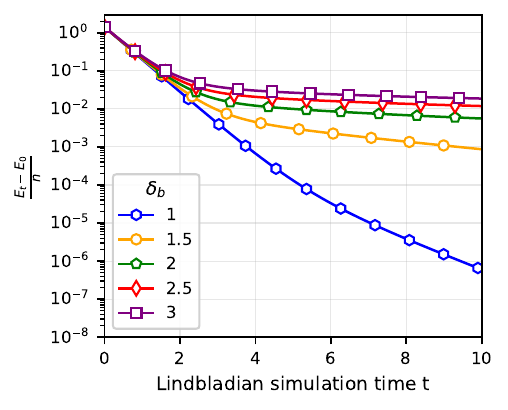}
    \end{minipage}
    
    \caption{Comparison of different filter functions with varying $\delta_b$ parameter. Left: Plot of localization error $\|K_a- K_a^\Omega\|$. On the horizontal axis, $\ell$ denotes the localization radius. The dotted red line is $\epsilon = 10^{-4}$. Right: Plot of the energy of the dynamics over time. We can see that larger values of $\delta_b$ allow a smaller localization radius, but results in substantially larger mixing time.}
    \label{fig:various-filter-results}
\end{figure}

In this section we describe our choice of filter function. Fig.~\ref{fig:filter-functions} shows the choice of filter function that is used in our numerical experiments. The choice of the filter function is given by Eq.~\ref{eq:numerical-filter} which is based off of the one in \cite{zhan_rapid_2025}. We tried several sets of parameters and chose the best set that give the fast mixing time while allowing small localization.

In particular, we tried different values of $\delta_b$, which controls the decay rate of the filter function in time (Fig.~\ref{fig:various-filter-plots}). As shown in Fig.~\ref{fig:various-filter-results} larger values of $\delta_b$ allow us to have a smaller localization radius, but results in a much larger mixing time.

% % fig: various LRB
% \begin{figure}
%   \centering
%   \includegraphics[width=0.8\linewidth]{fig/various_lrb.pdf}
%   \caption{}
%   \label{fig:various_lrb}
% \end{figure}

% % fig: various exact
% \begin{figure}
%   \centering
%   \includegraphics[width=0.8\linewidth]{fig/various_exact.pdf}
%   \caption{}
%   \label{fig:various_exact}
% \end{figure}

% \begin{figure}
%     \centering
%     \includegraphics[width=0.8\linewidth]{fig/various_lindblad.pdf}
%     \caption{Caption}
%     \label{fig:various_lindblad}
% \end{figure}

% \begin{figure}
%     \centering
%     \includegraphics[width=0.8\linewidth]{fig/various_filters.pdf}
%     \caption{Caption}
%     \label{fig:various_filters}
% \end{figure}

This filter function is derived in \cite{ding_single-ancilla_2024}, where the filter function is written as a product of Gevrey functions, defined in Definition \ref{def:gevrey}. This is formalized in the following assumption.
\begin{assumption}[Form of filter function]
    \label{assumption:filter}
    Choose $S_\omega = \bigO(\|H\|)$ such that $S_\omega > 2 \Delta$. Then the filter function $\hat f(\omega)$ can be written as
    \begin{equation}
        \hat f(\omega) = \hat u \left( \frac{\omega}{S_\omega}\right) \hat v \left(\frac{\omega}{\Delta}\right) \, 
    \end{equation}
    where $\hat u$ and $\hat v$ are $\nu$-Gevrey functions such that $\mathrm{supp}(\hat u) \subseteq [-1, 1]$ and $\hat u(\omega) = \Omega(1)$ when $\omega \in \left[-\frac{1}{2}, \frac{1}{2}\right]$, and 
    $\left \| \frac{d}{d\omega}\hat v\right\|_{L^1} = \bigO(1)$, $\mathrm{supp}(\hat v) \subseteq (-\infty, 0]$.
\end{assumption}
\begin{definition}[Gevrey class]
    \label{def:gevrey}
    \sloppy A function $\hat u$ is a $\nu$-Gevrey function if it belongs to the Gevrey class $\mathcal{G}_{A, B}^\nu (\R)$ for some $A, B > 0$ and $\nu \geq 1$ if
    \begin{equation}
        \sup_{\omega \in \R} \left| \frac{d^N}{d\omega^N} \hat u (\omega) \right| \leq AB^N N^{N\nu}
    \end{equation}
    for any $N \in \mathbb{N}$.
\end{definition}

Adopting this assumption on $\hat f$, the filter function in the time domain $f(t)$ will decay rapidly given by the following lemma, which is adapted from lemma 16 in \cite{ding_single-ancilla_2024}.
\begin{lemma}[Rapid decay of filter function]
    \label{lemma:rapid-decay-filter}
    If $\hat{f}$ satisfies Assumption~\ref{assumption:filter}, then its inverse Fourier transform $f(s)$ satisfies
    \begin{equation}
        \begin{split}
            |f(s)| &= \bigO(\|H\| \exp(-|s \Delta|^{1/\nu}))\\
            \|f(s)\|_{L^1} &=  \bigO( \log (\|H\| / \Delta)) \, .
        \end{split}
    \end{equation}
\end{lemma}

Choosing $\nu = 1$ allows us to choose $f$ that satisfies Assumption~\ref{assumption:rapid-decay-filter}. Note that \cite{ding_single-ancilla_2024} chose $\nu > 1$ for theoretical results in order for $\hat f$ to have compact support. However, in their numerical results they used a function with $\nu = 1$. It is necessary to do so in order to have rapid decay, but this results in $\hat f$ not having compact support but instead an exponentially decaying support. This means that the Nyquist-Shannon sampling theorem cannot be used as done in their work, but instead Poisson summation can be used which is still very effective.

\section{Mathematical tools}
\label{appendix:proofs}
This section details mathematical tools used to prove the main result in Section~\ref{sec:mainresult}. First, we state Holder's inequality on Schatten norms for linear operators.
% Holder's inequality for Schatten p-norms
\begin{lemma}[Holder's inequality for Schatten p-norms]
    \sloppy For $p \in [1, \infty]$ and $q$ such that $\frac{1}{p} + \frac{1}{q} = 1$, and linear operators $S, T$ we have $\|ST \|_1 \leq \|S\|_p\|T\|_q$.
    Thus for $p=\infty$ and $q = 1$, we have
    \begin{equation}
        \| X \rho Y \|_1 \leq \|X \|  \|\rho \|_1 \|Y \|  \, .
    \end{equation}
\end{lemma}

Holder's inequality can be used to prove the following two lemmas.
% approximate 1st term
\begin{lemma}
    \label{lemma:jump1}
    Let $\Phi_A$ denote the superoperator given by $\rho \mapsto \{A^\dagger A, \rho\}$. If $\|A - B\|  \leq \epsilon$, then $\|\Phi_A - \Phi_B\|_{1\to1} = \bigO(\epsilon)$.
\end{lemma}
\begin{proof}
    Let $E = A - B$. First note that for all $\rho \neq 0$,
    \begin{align}
        \|A^\dagger A \rho - B^\dagger B \rho \|_1 
        &= \| A^\dagger A \rho - (A - E)^\dagger (A- E)\rho \|_1 \\
        &= \| A^\dagger E \rho + E^\dagger A \rho - E^\dagger E \rho \|_1 \\
        &\leq \|A^\dagger E \rho \|_1 + \|E^\dagger A \rho \|_1 + \| E^\dagger E \rho \|_1 \\
        &\leq \|A^\dagger\|  \|E\|  \|\rho\|_1 + \|E^\dagger\|  \|A\|  \|\rho \|_1 
        + \|E^\dagger \|  \|E\|  \|\rho \|_1\\
        &\leq \|\rho \|_1 (2 \|A\|  \epsilon + \epsilon^2) \\
        &= \|\rho \|_1 \bigO(\epsilon).
    \end{align}

    Thus for all $\rho \neq 0$,
    \begin{equation}
        \frac{\| (\Phi_A - \Phi_B)[\rho] \|_1}{\|\rho \|_1} = \bigO (\epsilon) \, 
    \end{equation}
    so we have $\| (\Phi_A - \Phi_B) [\rho] \|_{1\to1} = \bigO(\epsilon)$.

\end{proof}

% approximate 2nd term
\begin{lemma}
    \label{lemma:jump2}
    Let $\Phi_A$ denote the superoperator given by $\rho \mapsto A\rho A^\dagger$. If $\| A - B\|  \leq \epsilon$ then $\| \Phi_A - \Phi_B \|_{1\to1} = \bigO(\epsilon)$. 
\end{lemma}
\begin{proof}
    Let $E = A- B$. Then for any $\rho \neq 0$, 
    \begin{align}
        \|A \rho A^\dagger - B\rho B^\dagger\|_1
        &= \|A\rho A^\dagger - (A-E)\rho (A-E)^\dagger \|_1 \\
        &= \|A\rho E^\dagger + E\rho A^\dagger - E\rho E^\dagger \|_1\\
        &\leq \|A\rho E^\dagger \|_1 + \|E\rho A^\dagger \|_1 + \|E\rho E^\dagger \|_1 \\
        &\leq \|A\|  \|\rho\|_1\|E\|  + \|E\|  \|\rho\|_1 \|A\|  
        + \|E\|  \|\rho\|_1 \|E\| 
    \end{align}
    Dividing by $\|\rho\|_1$ we get that
    \begin{align}
        \frac{\|(\Phi_A- \Phi_B) [\rho] \|_1}{\|\rho \|_1} &\leq 2 \|A\|  \|E\|  + \|E\| ^2 \\
        &\leq 2\|A\|  \epsilon + \epsilon^2 \, .
    \end{align}
    Since this holds for all $\rho \neq 0 $, we have that $\|\Phi_A - \Phi_B \|_{1\to1} \leq 2 \|A \|  \epsilon + \epsilon^2 = \bigO(\epsilon)$.

\end{proof}

The norm of each jump operator $K_a$ is bounded by a common value. This is shown in the following lemma.
\begin{lemma}
    For every jump operator $K_a$ defined as in Eq.~\ref{eq:rapid-jump}, we have $\|K_a \| \leq \|f\|_{L^1}$.
    \label{lemma:jump-operator-norm-bound}
\end{lemma}
\begin{proof}
    \begin{align}
        \|K_a \| &= \left\| \int_{-\infty}^\infty f(s) \, A_a(s) \, ds\right\|\\
        &\leq \int_{-\infty}^\infty |f(s)| \, \|A_a(s)\| \, ds\\
        &= \int_{-\infty}^\infty |f(s)| \, \|e^{iHs} A_a e^{-iHs}\| \, ds\\
        &= \int_{-\infty}^\infty |f(s)| \, \cancelto{1}{\|A_a\|} \, ds \\
        &= \| f \|_{L^1}
    \end{align}

    If $f$ follows the rapid decay assumption in Lemma~\ref{lemma:appx-jump}, then $\|f\|_{L^1} \leq 2C / \mu$.
\end{proof}

The following lemma is used to show a bound on the commutator of the superoperators $\|[\scriptH, \scriptK_a]\|_1$ and $\|[\scriptK_a, \scriptK_b]\|_1$. Here, $\|\cdot \|_1$ is shorthand for the superoperator norm $\|\cdot \|_{1\to1}$.

\begin{lemma}
    \label{lemma:commutator-bound}
    For $\scriptH$ defined as $\scriptH[\rho] = -i[H, \rho]$ and $\scriptK_a$ defined as $\scriptK_a [ \rho] = K_a\rho K_a^\dagger - \frac{1}{2}\{ K_a^\dagger K_a, \rho\}$, we have $\| [\scriptH, \scriptK_a] \|_1 \leq 8 \, \|H\|\, \|K_a\|^2 $ and $\| [\scriptK_a, \scriptK_b]\|_1 \leq 8\, \|K_a\|^2 \|K_b\|^2$
\end{lemma}
\begin{proof}
First find a bound on $\|\scriptH\|$ and $\|\scriptK_a\|$.
\begin{align}
    \|\scriptH[\rho]\|_1 &= \| - i [H, \rho ] \|_1\\
    &= \|H\rho - \rho H\|_1 \\
    &\leq \|H \rho \|_1 + \| \rho H \|_1 \\
    &\leq 2 \|H\| \, \|\rho\|_1 \, . \\
    \| \scriptK_a [\rho]\| _1 &\leq \left\|K_a \rho K_a^\dagger - \frac{1}{2}(K_a^\dagger K_a\rho + \rho K_a^\dagger K_a) \right\|_1 \\
    &= \|K_a \rho K_a^\dagger \|_1 + \frac{1}{2} \| K_a^\dagger K_a\rho \|_1 + \frac{1}{2}\|\rho K_a^\dagger K_a \|_1 \\
    &\leq 2\|K_a\|^2 \, \|\rho \|_1\ .
\end{align}
Thus $\|\scriptH\|_1 \leq 2\|H\|$ and $\|\scriptK_a \|_1 \leq 2\|K_a\|^2$. Then by Lemma~\ref{lemma:jump-operator-norm-bound},

\begin{align}
    \| [\scriptH, \scriptK_a] \|_1 &= \| \scriptH \scriptK_a - \scriptK_a \scriptH\|_1 \\
    &\leq 2 \, \| \scriptH \|_1 \, \|\scriptK_a \|_1\\
    &\leq 2 \cdot (2 \|H\|) \cdot (2\|K_a\|^2)\\
    &= 8\, \|H \| \, \|K_a\|^2 \, \\
    \|[\scriptK_a, \scriptK_b] \|_1 &= \| \scriptK_a \scriptK_b - \scriptK_b \scriptK_a \|_1 \\
    &\leq 2\, \|\scriptK_a \|_1 \, \|\scriptK_b \|_1\\
    &= 2  (2\|K_a\|^2) (2\|K_b\|^2)\\
    &= 8 \|K_a\|^2 \|K_b\|^2
\end{align}
\end{proof}

\begin{lemma}
    \label{lemma:duhamel-lindbladian-evol}
    For two Lindbladians $\Li, \Li^\Omega$ and any time $t>0$, we have
    \begin{equation}
        \left\| e^{\Li t} - e^{\Li^\Omega t} \right\|_{1\to1} \leq \left\|\Li - \Li^\Omega \right\|_{1\to1} t \, .
    \end{equation}
\end{lemma}
\begin{proof}
    Duhamel's identity (Eq.~43 in \cite{haber_notes_2018}) states that
    \begin{equation}
        e^A - e^B = \int_0^1 e^{sA}(A - B) e^{(1-s)B}\, ds\, .
    \end{equation}
    Thus
    \begin{align}
        \left\| e^{\Li t} - e^{\Li^\Omega t} \right\|_{1\to1}
        \leq & \int _0^1 \left\| e^{\Li s t} (\Li t - \Li^\Omega t)e^{\Li^\Omega (1-s) t} \right\|_{1\to1} \, ds \\
        \leq & \int_0^1 \left\| e^{\Li st}\right\|_{1\to1} \left\|\Li - \Li^\Omega \right\|_{1-1 }t \, \left\|e^{\Li^\Omega (1-s)t} \right\|_{1\to1} \, ds \\
        = & \left\|\Li - \Li^\Omega \right\|_{1\to1} t
    \end{align}
    where we used the fact that $\|e^{x\Li}\|_{1\to1} = 1$ because $e^{x\Li}$ is a CPTP map \cite{wolf_quantum_2012}.
    
\end{proof}

\section{Implementation of \texorpdfstring{$e^{-iH\tau}$}{e\^ -iH tau} and \texorpdfstring{$e^{-i\tilde K_a^\Omega \sqrt{\tau}}$}{e\^ -i tilde K\_a \^ Omega sqrt tau}}
\label{sec:gate-implementation}
In this section, we decompose the unitary operations of $e^{-iH\tau}$ and $e^{-i\tilde{K}_a^\Omega \sqrt{\tau}}$ into gates. Our goal is to simulate both to accuracy $\bigO(\tau^2)$ so that it does not grow faster than the Trotter error.

% Truncation and discretization
Because the $e^{-iH\tau}$ is more straightforward, we first address the implementation of $e^{-i\tilde K_a^\Omega \sqrt{\tau}}$. We build off of the method in \cite{ding_single-ancilla_2024}. Recall that $K_a^\Omega$ is given by
\begin{equation}
    K_a^\Omega = \int_{-\infty}^\infty f(s)\, A_a^\Omega (s) \, ds \, .
\end{equation}
To implement this operator, truncate it within some range $[s_{-M}, s_M]$ with $s_{-M} = -s_{M}$ and discretize it with time step $\Delta \tau$. This gives the discretized operator $K'$ given by
\begin{equation}
    (K_a^\Omega)' = \sum_{\ell = -M}^M f(s_\ell) \, A_a^\Omega(s_\ell) \, \Delta\tau\, 
    \label{eq:discretized-jump}
\end{equation}
where $s_\ell = \ell \Delta \tau$.
To ensure that the quadrature is sufficiently suppressed, we choose the truncation range $s_M$ and time step $\Delta \tau$ according to the following lemma. We see that this is where the decay parameters $C, \mu, \hat C, \hat \mu$ come into play. Notably, $\mu$ has a big effect on the truncation range $s_M$ while the other parameters have smaller effect.

\begin{lemma}
    \label{lemma:quadrature-error}
    Let $f$ satisfy Assumption ~\ref{assumption:rapid-decay-filter}. Let $r$ be the radius of $\supp ( H_a^\Omega)$. Then for any $\epsilon '> 0 $, we obtain $\|K_a^\Omega - (K_a^\Omega)'\| = \bigO(\epsilon')$ if
    \begin{align}
            s_M &= \Omega \left( \frac{1}{\mu} \log \left( \frac{C}{\mu \epsilon'}\right)\right)\\
            \Delta \tau &= \bigO \left( \frac{\hat \mu}{\hat \mu \|H_a^\Omega \| +  \log \left(\frac{\hat C}{\epsilon'}\right) + r^D}\right) \, .
    \end{align}
\end{lemma}
\begin{proof}
    Let
    \begin{equation}
        K_\infty' \coloneq \sum_{\ell \in \mathbb{Z}} f(s_\ell) A_a^\Omega (s_\ell) \Delta \tau \, .
    \end{equation}
    Then by triangle inequality,
    \begin{equation}
        \|K_a^\Omega - (K_a^\Omega)'\| \leq \underbrace{\|K_a^\Omega - K_\infty'\|}_{\text{discretization error}} + \underbrace{\|K_\infty' - (K_a^\Omega)'\|}_{\text{truncation error}} \, .
    \end{equation}
    It suffices to bound both terms by $\bigO(\epsilon')$.
    
    To analyze the discretization error, write $K_a^\Omega$ and $K_\infty'$ as
    \begin{equation}
    \begin{split}
        K_a^\Omega &= \sum_{i,j} \hat f(\lambda_i - \lambda_j) \, |i\rangle \langle i|A_a|j\rangle \langle j| \,\\
        K_\infty' &= \sum_{i,j} \left( \sum_{\ell=-\infty}^\infty f(s_\ell) e^{-i(\lambda_i - \lambda_j) s_\ell} \Delta \tau\right) \, |i\rangle \langle i|A_a|j\rangle \langle j|
    \end{split}
    \end{equation}
    where $\lambda_i, \lambda_j$ denotes eigenvalues of $H_a^\Omega$. 
    From here we can see that $\hat f(\lambda_i - \lambda_j)$ is approximated by the Poisson summation, which is expressed as
    \begin{equation}
        \Delta \tau \sum_{\ell \in \mathbb{Z}} f(s_\ell) e^{-i\omega s_\ell} = \sum_{k\in \mathbb{Z}} \hat f\left( \omega + \frac{2\pi k}{\Delta \tau}\right)\, 
        \label{eq:poisson-sum}
    \end{equation}
    where $\omega$ represents $\lambda_i - \lambda_j$, the Bohr frequencies. This implies that the approximation error to $\hat f(\omega)$ is
    \begin{equation}
        E(\omega) \coloneq \left|\Delta \tau \sum_{\ell \in \mathbb{Z}} f(s_\ell) e^{-i\omega s_\ell} - \hat f(\omega) \right| \leq \sum_{k\neq 0} \left| \hat f \left( \omega + \frac{2\pi k}{\Delta \tau} \right)\right| \, .
    \end{equation}
    By property of spectral norm, $|\omega| \leq 2\|H_a^\Omega \|$. Furthermore, by assumption that $|\hat f (\xi) |\leq \hat C e^{-\hat \mu |\xi|}$ we have
    \begin{align}
        \left| \hat f \left( \omega + \frac{2\pi k}{\Delta \tau} \right)\right| &\leq \hat C \exp \left(- \hat \mu \left| \omega + \frac{2\pi k}{\Delta \tau}\right| \right)\, \\
        &\leq \hat C \exp \left(-\hat \mu \left( \frac{2\pi |k|}{\Delta \tau} - |\omega|\right)\right) \\
        &= \bigO\left( \hat C\,  e^{2\hat \mu \|H_a^\Omega \|}\, \exp \left(- \frac{2\pi \hat \mu |k|}{\Delta \tau}\right)\right) \, .
    \end{align}
    Thus,
    \begin{align}
        E(\omega) &\leq 2\hat C e^{2\hat \mu \|H_a^\Omega \|}\sum_{k =  1}^\infty \exp \left(- \frac{2\pi \hat \mu k}{\Delta \tau}\right)\\
        &= 2\hat C e^{2\hat \mu \|H_a^\Omega \|} \frac{e^{- \frac{2\pi \hat \mu}{\Delta \tau}}}{1-  e^{- \frac{2\pi \hat \mu }{\Delta \tau}}}\\
        &= \hat 2 Ce^{2\hat \mu \|H_a^\Omega \|} e^{- \frac{2\pi \hat \mu}{\Delta \tau}}
    \end{align}
    where we used the geometric series formula.
    
    Using the fact that the total number of Bohr frequencies of $H_a^\Omega$ is upper bounded by $4^{\bigO(r^D)}$, the total discretization error is
    \begin{align}
        \|K_a^\Omega - K_\infty'\| &\leq \sum_{i,j}  E(\omega) \|A_a\| \\
        &\leq 4^{\bigO(r^D)} E(\omega) \\
        &\leq 4^{\bigO(r^D)}\, 2\hat Ce^{2\hat \mu \|H_a^\Omega \|}\, e^{- \frac{2\pi \hat \mu}{\Delta \tau}} \, .
    \end{align}
    To bound this error by $\bigO(\epsilon')$, it suffices to choose
    \begin{equation}
        \Delta \tau < \frac{2\pi \hat \mu}{2\hat \mu \|H_a^\Omega\| + \log \frac{2\hat C}{\epsilon'} + \bigO(r^D)} \, .
    \end{equation}

    The truncation error is
    \begin{align}
        \| K_\infty' - (K_a^\Omega)'\| &\leq \Delta \tau \|A_a\| \sum_{|\ell| > M } |f(s_\ell)| \\
        &\leq 2 C \Delta \tau \sum_{\ell = M+1}^\infty e^{-\mu s_\ell} \\
        &\leq 2 C \int_{s_M}^\infty e^{- \mu s} \, ds\\
        &= \frac{2  C}{ \mu} e^{- \mu s_M} \, .
    \end{align}
    Thus to bound this error by $\bigO(\epsilon')$, it suffices to choose
    \begin{equation}
        s_M > \frac{1}{ \mu} \log \frac{2C}{ \mu \epsilon'}\, .
    \end{equation}
\end{proof}
% which is adapted from Lemma~24 in \cite{ding_single-ancilla_2024}.
% \begin{lemma}[Quadrature error]
%     \label{lemma:quadrature-error}
%     Let $\hat f$ satisfy Assumption~\ref{assumption:filter}. Then for any $\epsilon > 0$, we obtain $\|K_a^\Omega - (K_a^\Omega)' \| = \bigO(\epsilon)$ if
%     \begin{align}
%             s_M &= \Omega \left( \frac{1}{\Delta}\mathrm{polylog}\left( \frac{\|H\|}{\Delta \epsilon} \right)\right)\\
%             \Delta \tau &= \bigO \left( \frac{1}{\| H_a^\Omega \|}\right)\, .
%     \end{align}
% \end{lemma}

% Dilation
The dilated jump operator $\tilde K_a^\Omega$ is given by
\begin{equation}
    \tilde K_a^\Omega = \begin{pmatrix}
        0 & (K_a^\Omega)^\dagger \\
        K_a^\Omega & 0 
    \end{pmatrix}
\end{equation}
thus the discretized operator $(\tilde K_a^\Omega)'$ is
\begin{equation}
    (\tilde K_a^\Omega)' = \sum_{\ell =-M}^M \begin{pmatrix}
        0 & f^*(s_\ell) A_a^\Omega(s_\ell)\Delta \tau \\
        f(s_\ell) A_a^\Omega(s_\ell)\Delta \tau & 0
    \end{pmatrix}\, .
\end{equation}
Define $\sigma_\ell \coloneq \Delta\tau \, (\sigma_x\, \mathrm{Re}\, f(s_\ell) \, + \sigma_y\, \mathrm{Im}\, f(s_\ell))$ where $\sigma_x$ and  $\sigma_y$ are the Pauli $X$ and $Y$ matrices, respectively. Then $(\tilde K_a^\Omega)'$ can be written as
\begin{equation}
    (\tilde K_a^\Omega)' = \sum_{\ell = -M}^M \tilde A_a^\Omega (s_\ell) \, 
\end{equation}
where
\begin{equation}\tilde A_a^\Omega (s_\ell) \coloneq \sigma_\ell \otimes A_a^\Omega(s_\ell),.
\end{equation}
Thus to implement $e^{-i\tilde K_a^\Omega \sqrt{\tau}}$, we need to simulate
\begin{equation}
    e^{-i(\tilde K_a^\Omega)' \sqrt{\tau}} = \exp \left(-i \sqrt{\tau} \sum_{\ell = -M}^M \tilde A_a^\Omega (s_\ell) \right) \, .
\end{equation}

% Product formula
Each term in the summation contains a Heisenberg evolution $A_a^\Omega(s_\ell)$. To pull out its left and right $e^{-iH^\Omega s_\ell}$ evolutions, use a product formula to remove the summation out of the exponent. Following \cite{ding_single-ancilla_2024}, we use a second-order Trotter formula so that errors are sufficiently suppressed. This is given by 
\begin{equation}
    e^{-i(\tilde K_a^\Omega)' \sqrt{\tau}} \approx \prod_\ell^\to  \exp\left( -i\frac{\sqrt{\tau}}{2} \tilde A_a^\Omega(s_\ell)\right) \prod_\ell^\gets \exp \left( -i\frac{\sqrt{\tau}}{2} \tilde A_a^\Omega(s_\ell)\right) \, 
    \label{eq:second-order-trotter}
\end{equation}
where $\prod_\ell^\to, \prod_\ell^{\gets}$ denote ordered products in which the index $\ell$ increases and decreases from left to right, respectively. The error incurred from this approximation is analyzed in Lemma~\ref{lemma:stinespring-dilation-gates}.

% Pulling left and right evolutions out of exponent
After decomposing into a product formula, we can use the following identity to pull the left and right evolutions out of the exponent:
\begin{equation}
\begin{split}
    &\quad \exp\left( -i\frac{\sqrt{\tau}}{2} \tilde A_a^\Omega (s_\ell)\right)\\ &= \exp\left( -i\frac{\sqrt{\tau}}{2} \left(I \otimes e^{iH_a^\Omega s_\ell}\right)(\sigma_\ell \otimes A_a)\left(I \otimes e^{-iH_a^\Omega s_\ell}\right)\right) \\
    &= \left(I \otimes e^{iH_a^\Omega s_\ell}\right) \exp\Bigl( -i\frac{\sqrt{\tau}}{2} (\underbrace{\sigma_\ell \otimes A_a}_{\coloneq \tilde A_a^\ell \, })\Bigr) \left(I \otimes e^{-iH_a^\Omega s_\ell}\right) \, .
\end{split}
\label{eq:heisenberg-evol-pullout}
\end{equation}
This identity can be shown by taking the Taylor expansion definition of matrix exponential.

% Cancellations
Using the identity from Eq.~\ref{eq:heisenberg-evol-pullout} into the Trotter formula in Eq.~\ref{eq:second-order-trotter}, we obtain a product of Hamiltonian evolutions. Furthermore, there exists cancellations in those evolutions, e.g.
\begin{equation}
(I \otimes e^{-iH_a^\Omega s_\ell}) (I \otimes e^{iH_a^\Omega s_{\ell + 1}}) = (I \otimes e^{iH_a^\Omega \Delta \tau})\, .
\end{equation}
Thus, the product in Eq.~\ref{eq:second-order-trotter} becomes
\begin{equation}
    e^{-iH_a^\Omega s_M} \left(\prod_\ell^\to e^{-i \frac{\sqrt{\tau}}{2} \tilde A_a^\ell} e^{iH_a^\Omega \Delta \tau} \right)
    \left(\prod_\ell^\gets e^{-iH_a^\Omega \Delta \tau} e^{-i\frac{\sqrt{\tau}}{2} \tilde A_a^\ell}\right) e^{iH_a^\Omega s_M} \, .
    \label{eq:final-product}
\end{equation}
We use the shorthand notation $e^{-iH_a^\Omega s}$ to denote $(I \otimes e^{-iH_a^\Omega s})$.
To simulate each term in this product, we use qubitization. Specifically, we use the following lemma from \cite{low_hamiltonian_2019}.

% Because $H_a^\Omega$ is restricted to a local set $\Omega$, it has much less terms than the original $H$. Furthermore, $H_a^\Omega$ acts on $k$ qubits while $\tilde A_a^\ell$ acts on two: one ancilla and one system qubit.

% To simulate $e^{-i\tilde K_a^\Omega \sqrt{\tau}}$ up to accuracy $\bigO(\tau^2)$, it suffices to simulate $e^{-iH_a^\Omega s_M}$,
% $e^{-iH_a^\Omega \Delta \tau}, $ and $e^{-i \frac{\sqrt{\tau}}{2}\tilde A_a^\ell}$ up to accuracy $\bigO(\tau^2 / M)$ because error adds up linearly. This can be done by using qubitization \cite{low_hamiltonian_2019} assuming access the corresponding Prepare and Select oracles for each Hamiltonian. In particular, we require the following oracles
% \begin{equation}
% \label{eq:oracle-defined}
% \begin{split}
%     &\mathrm{Prep}_{H_a^{\Omega}}\, , \quad \mathrm{Sel}_{H_a^{\Omega}} \qquad \forall a \in [n]\\
%     &\mathrm{Prep}_{\tilde A_a^\ell}\, , \, \quad \mathrm{Sel}_{\tilde A_a^\ell} \qquad \forall \ell \in [-M , M], \, \forall a \in [n]\, .
% \end{split}
% \end{equation}
% where $\Omega_j$ is the local set of qubits centered around qubit $j$. The following lemma is adapted from Corollary 16 in \cite{low_hamiltonian_2019}.

\begin{lemma}[Hamiltonian Simulation via Qubitization]
\label{lemma:qubitization}
% Prep and Sel are defined as follows.
Given a matrix $A$ written in a linear combination of unitaries (LCU) as $A = \sum_k \alpha_k U_k$,
% then
% \begin{equation}
%     \begin{split}
%         \mathrm{Prep}_{A}\, |0\rangle  &= \sum_k \sqrt{\frac{|\alpha_k|}{\lambda}}\, U_k\\
%         \mathrm{Sel}_A\, |k\rangle |\psi\rangle &= |k\rangle\, U_k |\psi\rangle
%     \end{split}
% \end{equation}
% where $\lambda = \sum_k |\alpha_k|$.
$e^{-iAt}$ can be simulated to accuracy $\epsilon$ using qubitization with $\bigO(\alpha t + \log (1/\epsilon))$ many queries to Prep and Sel, where $\alpha = \sum_k |\alpha_k|$
\end{lemma}

For a geometrically $k$-local Hamiltonian $H = \sum_{X\subseteq\Lambda} h_X$, we have that each local term $h_X$ acts nontrivially only on a region of geometrically $k$-local qubits. Thus, it can be represented in a linear combination of $3^{k^D}$ Pauli operators, which is constant because $k$ is constant. Together with the bounded norm assumption that $\|h_X\| = \bigO(1)$, the 1-norm of the LCU coefficients of $h_X$ is constant. Under this assumption, the 1-norm of the LCU coefficients grows linearly with the size of the Hamiltonian.

The implementation of Prep and Sel oracles require a number of gates that scales linearly with the number of LCU terms \cite{low_hamiltonian_2019}. Thus, the implementation of Prep and Sel for $h_X$ utilizes $\bigO(1)$ gates, and since there are linearly many $h_X$, that for $H$ utilizes $\bigO(n)$ gates. In other words, under our assumptions, the number of gates to implement the oracles scale linearly with the size of the Hamiltonian.

\begin{lemma}
\label{lemma:gates-h}
     $e^{-ih_X\tau}$ can be approximated to accuracy $\epsilon$ in trace norm using $\bigO\left(\log \left( \frac{1}{\epsilon}\right)\right)$ queries and gates.
\end{lemma}
\begin{proof}
    Using qubitization in Lemma~\ref{lemma:qubitization}, $e^{-ih_X\tau}$ can be approximated to accuracy $\epsilon$ using $\bigO(\alpha \tau + \log(1/\epsilon)$ queries to $\mathrm{Prep}_{h_X}$ and $\mathrm{Sel}_{h_X}$, where $\alpha$ is the 1-norm of the LCU coefficients of $h_X$. By assumption, $\alpha = \bigO(1)$ and $\tau$ is chosen by Eq.~\ref{eq:tau-formula}, hence $\alpha \tau$ term is negligible. Overall we use $\bigO(\log (1/\epsilon))$ queries, each of which uses $\bigO(1)$ gates.

\end{proof}

The following lemma analyzes the error associated with implementing the product in Eq.~\ref{eq:final-product} to implement $e^{-i\tilde K_a^\Omega \sqrt{\tau}}$.

\begin{lemma}
\label{lemma:n-gates} 
    Let $r$ be the radius of $\supp (H_a^\Omega)$ and $J$ be a constant bounding the interaction strength, i.e. $\|h_X\| < J$ for all $X$. If $f$ satisfies Assumption~\ref{assumption:rapid-decay-filter} and \ref{assumption:rapid-decay-filter-fourier} then there exists an implementation using
    \begin{equation}
        N_{\mathrm{queries}} = \bigO\left(\left( \frac{J}{\mu} + \frac{1}{\mu \hat \mu}\right) r^{D} \polylog \left( \frac{1}{ \epsilon'\epsilon}\right) \right)\, 
    \end{equation}
    queries to $\Prep_{H_a^\Omega}$ and $\Sel_{H_a^\Omega}$ and
    \begin{equation}
        N_{\mathrm{gates}} = \bigO\left(\left( \frac{J}{\mu} + \frac{1}{\mu \hat \mu}\right) r^{2D} \polylog \left( \frac{1}{ \epsilon'\epsilon}\right) \right)\, 
    \end{equation}
    gates that approximates the product in Eq.~\ref{eq:final-product} to accuracy $\epsilon$ in trace distance, where $\epsilon'$ is the quadrature error in Lemma~\ref{lemma:quadrature-error}. Here, we hide logarithmic factors in $C, \hat C, \mu, \hat \mu$.
\end{lemma}

\begin{proof}
    We proceed term by term and utilize the fact that by geometric locality, $\alpha_a^\Omega = \bigO(Jr^D) $.

    To simulate $e^{-iH_a^\Omega s_M}$ to accuracy $\bigO(\epsilon)$, we require $\bigO \left(\alpha_a^\Omega s_M + \log \left( \frac{1}{\epsilon}\right)\right)$ queries to $\mathrm{Prep}_{H_a^\Omega}$ and $\mathrm{Sel}_{H_a^\Omega}$.  Using Lemma~\ref{lemma:quadrature-error}, we choose $s_M = \Theta \left( \frac{1}{\mu}\mathrm{log}\left( \frac{C}{\mu \epsilon'} \right)\right)$.
    Thus $ \bigO \left( \frac{Jr^{D}}{\mu}\mathrm{log}\left( \frac{C}{\mu \epsilon' \epsilon} \right)\right)$ queries are sufficient to simulate $e^{-iH_a^\Omega s_M}$.

    To simulate $e^{-i H_a^\Omega \Delta \tau}$ to accuracy $\bigO(\epsilon / M)$ we require $\bigO \left(\alpha_a^\Omega \Delta\tau  + \log \left( \frac{M}{\epsilon}\right)\right)$ queries to $\mathrm{Prep}_{H_a^\Omega}$ and $\mathrm{Sel}_{H_a^\Omega}$. We choose $\Delta \tau$ according to
    \begin{equation}
        \Delta \tau = \bigO \left( \frac{\hat \mu}{\hat \mu \alpha_a^\Omega +  \log \left(\frac{\hat C}{\epsilon'}\right) + r^D}\right)
    \end{equation}
    which also satisfies the condition for Lemma~\ref{lemma:quadrature-error}. This implies that the $\alpha_a^\Omega \Delta \tau$ term is negligible. Thus a total of $\bigO(\log (M/\epsilon))$ gates are used to simulate $e^{-i H_a^\Omega \Delta \tau}$.

    $e^{-i \frac{\sqrt{\tau}}{2} \tilde A_a^\ell}$ can be simulated exactly using a constant number of CNOT and single qubit rotation gates. This is because $\tilde A_a^\ell$ is of the form $(\alpha X + \beta Y) \otimes X$ where $\alpha$ and $\beta$ are real values that depend on $f(s_\ell)$. Thus its time evolution can be implemented by classically computing $\alpha$ and $\beta$ then using them as rotation angles in the circuit.

    Overall, to simulate $e^{-i\tilde K_a^\Omega \sqrt{\tau}}$, we need $\bigO(1)$ applications of $e^{-iH_a^\Omega s_M}$, $\bigO(M)$ applications of $e^{-i H_a^\Omega \Delta \tau}$, and $\bigO(M)$ applications of $e^{-i \frac{\sqrt{\tau}}{2} \tilde A_a^\ell}$. This gives a total of
    \begin{equation}
    \bigO \left(\frac{Jr^{D}}{\mu}\mathrm{log}\left( \frac{C}{\mu \epsilon' \epsilon} \right) + M \log \left(\frac{M}{\epsilon}\right)\right)    
    \end{equation}
    queries. $M$ is given by
    \begin{equation}
    M = \Theta \left(\frac{s_M}{\Delta\tau} \right) = \Theta \left( \frac{1}{\mu \hat \mu}\log \left(\frac{C}{\mu \epsilon'} \right)  \left( \hat \mu \alpha_a^\Omega + \log \left(\frac{\hat C}{\epsilon'}\right) + r^D\right)\right)\, .
    \end{equation}
    The number of queries used is
    \begin{equation}
        N_{\mathrm{queries}} = \bigO\left(\left( \frac{J}{\mu} + \frac{1}{\mu \hat \mu}  + \frac{1}{\mu \hat \mu\,  r^{D}} \log \left(\frac{\hat C}{\epsilon'}\right)\right) r^{D} \polylog \left( \frac{C}{\mu \epsilon' \epsilon}\right) \right)\, .
    \end{equation}
    Hiding logarithmic factors in $C, \mu, \hat C, \hat \mu$, we get the desired result. Since each query uses $\bigO(r^D)$ gates, we obtain $N_{\mathrm{gates}}$ as desired.

\end{proof}

Finally, the following lemma analyzes the error when implementing the Stinespring dilation channel $\Eo_a[\rho]$ which introduces an ancilla and traces it out. With this, we can approximate the entire jump operator channel $\scriptK_a$. 

\begin{lemma}
    \label{lemma:stinespring-dilation-gates}
    Let $\epsilon'$ be the quadrature error given by Lemma~\ref{lemma:quadrature-error}. For any $\tau >0$, the channel $\rho \mapsto \Eo_a [\rho]$ can be approximated to accuracy
    \begin{equation}
        \bigO \left(\|f\|_{L^1} \epsilon' \tau \right) + \bigO \left( \tau^2  \|f\|_{L^1}^4\right)
    \end{equation}
    in superoperator 1-1 norm with
    \begin{equation}
        N_{\mathrm{queries}} = \bigO\left(\left( \frac{J}{\mu} + \frac{1}{\mu \hat \mu}\right) r^{D} \polylog \left( \frac{1}{ \epsilon'\tau }\right) \right)\, 
    \end{equation}
    queries to $\Prep_{H_a^\Omega}, \Sel_{H_a^\Omega}$ and
    \begin{equation}
        N_{\mathrm{gates}} = \bigO\left(\left( \frac{J}{\mu} + \frac{1}{\mu \hat \mu}\right) r^{2D} \polylog \left( \frac{1}{ \epsilon'\tau }\right) \right)\, 
    \end{equation}
    gates. Here, we hide logarithmic factors in $C, \hat C, \mu, \hat \mu$.
\end{lemma}
\begin{proof}
    Let $\rho_a \coloneq |0\rangle \langle 0| \otimes \rho$. Let $U \coloneq e^{-i\tilde K_a^\Omega \sqrt{\tau}}$ and let its discretized approximation be $U' \coloneq e^{-i(\tilde K_a^\Omega)' \sqrt{\tau}}$. Then the total error for the gate implementation of $\Eo_a$ is given by
    \begin{align}
        \|\Eo_a[\rho] - \Tr_a (R'_{\Delta \tau} \, \rho_a\, (R'_{\Delta \tau})^\dagger) \|_1 &= \|\Tr_a (U \rho_a U^\dagger) - \Tr_a (U' \rho_a (U')^\dagger )\|_1 \\
        & + \| \Tr_a (U' \rho_a (U')^\dagger ) - \Tr_a (R_{\Delta \tau} \rho_a R_{\Delta \tau}^\dagger) \|_1  \\
        & + \| \Tr_a (R_{\Delta \tau} \rho_a R_{\Delta \tau}^\dagger) - \Tr_a  (R_{\Delta \tau}' \rho_a (R_{\Delta \tau}')^\dagger)\|_1
    \end{align}
    where $R_{\Delta \tau}$ is the product given in Eq.~\ref{eq:final-product} and $R_{\Delta \tau}'$ is its gate implementation according to Lemma~\ref{lemma:n-gates}. This splits the error into three terms: the first term is the discretization error, the second is Trotterization error, and the third is gate implementation error.

    Note that the first two terms are bounded not with standard channel bounds, but rather the property of partial trace and that $\tilde K_a^\Omega$ is off-diagonal in the ancilla basis. More detail can be found in \cite{ding_single-ancilla_2024}.

    The first term is bounded according to Eq.~(D8) in \cite{ding_single-ancilla_2024}, given by
    \begin{equation}
        \bigO \left( \left\|\tilde K_a^\Omega - (\tilde K_a^\Omega)'\right\| \left(\|\tilde K_a^\Omega \| + \|(\tilde K_a^\Omega )'\|\right ) \tau \right) = \bigO \left(\|f\|_{L^1} \epsilon' \tau \right)\, .
    \end{equation}

    The second term is bounded according to Eq.~(D9) in \cite{ding_single-ancilla_2024}, given by
    \begin{equation}
        \bigO \left( \left( \sqrt{\tau} \sum_{\ell} \left\| \tilde A_a^\Omega (s_\ell) \right\| \right)^4 \right) = \bigO \left( \tau^2  \|f\|_{L^1}^4\right) \, .
    \end{equation}

    With sufficiently many gates, the third term can be suppressed to an arbitrarily low error by Lemma~\ref{lemma:n-gates}. In this case, it suffices to suppress it to an error $\tau^2$, obtaining the number of queries and gates stated.

\end{proof}

\end{document}